\documentclass[sigplan]{acmart}

\usepackage{amssymb}

\usepackage{graphicx}
\usepackage{subcaption}
\usepackage{multirow}
\usepackage{amsfonts}
\usepackage{amsthm}
\usepackage[T1]{fontenc}
\usepackage{enumitem}
\usepackage{algorithm,algorithmicx,algpseudocode}

\newtheorem{assumption}{Assumption}
\newcommand{\E}{\mathbb{E}}
\newcommand{\R}{\mathbb{R}}

\AtBeginDocument{%
  \providecommand\BibTeX{{%
    \normalfont B\kern-0.5em{\scshape i\kern-0.25em b}\kern-0.8em\TeX}}}

\copyrightyear{2020}
\acmYear{2020}
\setcopyright{none}
\acmConference[PPoPP '20]{25th ACM SIGPLAN Symposium on Principles and Practice of Parallel Programming}{February 22--26, 2020}{San Diego, CA, USA}
\acmBooktitle{25th ACM SIGPLAN Symposium on Principles and Practice of Parallel Programming (PPoPP '20), February 22--26, 2020, San Diego, CA, USA}
\acmPrice{15.00}
\acmDOI{10.1145/3332466.3374528}
\acmISBN{978-1-4503-6818-6/20/02}

\begin{document}

\title{Taming Unbalanced Training Workloads in Deep Learning with Partial Collective Operations}


\author{Shigang Li}
\affiliation{
Department of Computer Science\\
ETH Zurich
}
\email{shigangli.cs@gmail.com}

\author{Tal Ben-Nun}
\affiliation{
Department of Computer Science\\
ETH Zurich
}
\email{talbn@inf.ethz.ch} 

\author{Salvatore Di Girolamo}
\affiliation{
Department of Computer Science\\
ETH Zurich
}
\email{salvatore.digirolamo@inf.ethz.ch} 

\author{Dan Alistarh}
\affiliation{
IST Austria
}
\email{dan.alistarh@ist.ac.at} 

\author{Torsten Hoefler}
\affiliation{
Department of Computer Science\\
ETH Zurich
}
\email{torsten.hoefler@inf.ethz.ch} 
\renewcommand{\shortauthors}{Shigang Li, et al.}

\begin{abstract}
Load imbalance pervasively exists in distributed deep learning training systems, either caused by the inherent imbalance in learned tasks or by the system itself. Traditional synchronous Stochastic Gradient Descent (SGD) achieves good accuracy for a wide variety of tasks, but relies on global synchronization to accumulate the gradients at every training step. In this paper, we propose \textit{eager-SGD}, which relaxes the global synchronization for decentralized accumulation. To implement eager-SGD, we propose to use two partial collectives: solo and majority. With solo allreduce, the faster processes contribute their gradients eagerly without waiting for the slower processes, whereas with majority allreduce, at least half of the participants must contribute gradients before continuing, all without using a central parameter server. We theoretically prove the convergence of the algorithms and describe the partial collectives in detail. Experiments are conducted on a variety of neural networks and datasets. The results on load-imbalanced environments show that eager-SGD achieves 2.64 $\times$ speedup (ResNet-50 on ImageNet) over the asynchronous centralized SGD, and achieves 1.29 $\times$ speedup (ResNet-50 on ImageNet) and 1.27$\times$ speedup (LSTM on UCF101) over the state-of-the-art synchronous decentralized SGDs, without losing accuracy.
\end{abstract}




\maketitle

\section{Motivation}

Deep learning models are on a steep trajectory to becoming the most
important workload on parallel and distributed computer systems.
Early convolutional networks demonstrated groundbreaking successes in
computer vision, ranging from image classification to object
detection~\cite{densenet,vgg}.
More recent developments in recurrent and transformer networks enable
impressive results in video classification, natural language processing for machine
translation, question answering, text comprehension, and synthetic text generation. The latter models contain more than 1.5
billion parameters and take weeks to train~\cite{bert,gpt-2}. Other
demanding neural networks are trained on the largest supercomputers to
achieve scientific breakthroughs~\cite{cosmoflow,climate-gb}.
Furthermore, the models are growing exponentially in size, OpenAI is
predicting a 10x growth each year~\cite{openai} potentially leading to artificial general intelligence. 
In order to support this development, optimizing the training procedure
is most important.

The training procedure of deep learning is highly parallel but dominated
by communication~\cite{dl-survey}.
Most parallel training schemes use data parallelism where full
models are trained with parts of the dataset and parameters are
synchronized at the end of each iteration. 
The total size of allreduce grows with the model size, which ranges from
a few megabytes~\cite{densenet} to several gigabytes~\cite{gpt-2} and
grows quickly. 
The allreduce operation is not atomic and it can be split into
layer-wise reductions, which can easily be overlapped with the layer
computation using non-blocking
collectives~\cite{hoefler2007implementation,panda-learning-paper}.
Yet, the optimal scaling of an allreduce of size $S$ is at best
$\mathcal{O}\left(\log P + S\right)$ in $P$
processes~\cite{allred-bounds-yuan,moor-colls-survey,sparcml}. Thus,
growing process counts will reduce the parallel efficiency and
eventually make the reduction a scaling bottleneck.

The communication aspects of deep learning have been investigated in
many different contexts~\cite{sergeev2018horovod,sparcml}, see the survey for an
overview~\cite{dl-survey}. In this work, we identify load imbalance as an
additional barrier to scalability. When some processes finish the
computation later than others, all processes will wait for the last one
at the blocking allreduce function.
Load imbalance can be caused by the system itself, for example, when
training on multi-tenant cloud systems~\cite{schad2010runtime, iosup2011performance, jackson2010performance} or by system or network
noise~\cite{network-noise,system-noise} in high-performance machines. 
A second, and more prominent cause of imbalance is inherent imbalance in
the computation that causes varying load across different processes.
While noise from the system is generally low on well-maintained HPC
machines~\cite{system-noise}, the inherent load imbalance of the
training workloads cannot easily be avoided.
Natural language processing tasks have sentences of highly
varying length while video processing tasks have videos with different number of frames. For example, the training dataset of UCF101~\cite{soomro2012ucf101} contains videos that range from 29 to 1,776 frames.

Several researchers have shown that the training process itself is quite
robust with respect to bounded errors. In fact, data augmentations such as Cutout~\cite{cutout} and Dropout~\cite{ba2013adaptive} introduce random errors and omissions into the training process to improve generalization properties. Several packages take advantage of this
robustness and employ three techniques in tandem: (1) communicated weights
are quantized to more compact number representations~\cite{seide20141,strom2015scalable}, (2)
only the most significant weights are sent during each
allreduce~\cite{sparcml,alistarh-nips18}, and (3) updates are only sent
to limited (random) neighborhoods using gossip algorithms~\cite{asyncring}.
We propose to exploit this robustness in a new way: we perform the
allreduce eagerly in that we ignore the input gradients of processes that come
late in order to not delay all processes. 
The communication partners are selected based on their workload (which
can be randomized) and the allreduce itself is performed with
high-performance reduction topologies~\cite{moor-colls-survey} in logarithmic
depth.
We call our method \emph{eager Stochastic Gradient Decent (eager-SGD)}, as a counterpart to synchronous SGD (synch-SGD)~\cite{deep500,sergeev2018horovod,panda-learning-paper}. Fig.~\ref{eagerSGD} shows the difference between synch-SGD and eager-SGD.

\begin{figure}[t]
\centering\includegraphics[width=.95\linewidth]{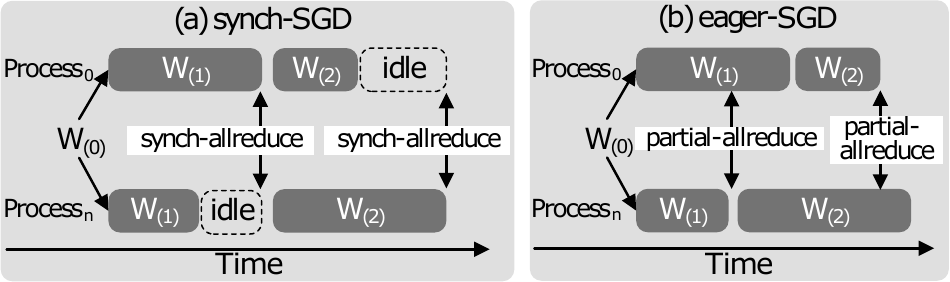}
\caption{\label{eagerSGD} Synch-SGD vs eager-SGD under load imbalance. w$_{(t)}$ are the weights in training step $t$.}
\end{figure}

Specifically, we propose to relax the allreduce operation to \emph{partial collectives} in eager-SGD. A partial collective is an asynchronous operation where a subset of the 
processes can trigger and complete the collective operation. Absentee processes follow a predefined protocol, such as contributing potentially outdated data. We define two partial collectives --- \textit{solo allreduce}, a wait-free operation that one process triggers; and \textit{majority allreduce}, in which the majority  must participate.

Our theoretical analysis shows that solo allreduce does not guarantee bounded error, as necessary in SGD, yet empirically converges in cases of moderate load imbalance. Majority allreduce is proven to bound the error, but is not completely wait-free. The statistical guarantee, however, is sufficient to both train deep neural networks and avoid the delays. 

We show that solo and majority collectives are suitable for different cases, depending on load imbalance severity.

Our main contributions are:
\begin{itemize}
  \item A detailed analysis of workload imbalance in deep learning training.
  \item Definition and implementation of partial collectives, specifically majority and solo allreduce.
  \item Eager-SGD for asynchronous decentralized distributed training of neural networks with proof of convergence.
  \item An experimental study of convergence and training speed for multiple networks, achieving 1.27$\times$ speedup over synchronous SGD on a video classification task without losing accuracy.
\end{itemize}

\section{Load-Imbalance in Deep Learning}
Load imbalance widely exists in the training of deep learning models, which can be caused by either the applications or the system itself~\cite{schad2010runtime,iosup2011performance,jackson2010performance,network-noise,system-noise}.

\subsection{Video Processing}
\label{imbvideomodel}

Long short-term memory (LSTM)~\cite{hochreiter1997long} is a type of unit cell in Recurrent Neural Networks (RNN). In video classification tasks, LSTMs are used~\cite{Ng_2015_CVPR,donahue14,ballas15} to process a sequence of frames for a video as input (optionally following convolutional neural networks that preprocess the images to features), and output a probability distribution over a set of classes. Due to the recurrent structure of the network, the computational overhead is proportional to the number of frames in the input video. 

Fig.~\ref{videolength} shows the video length distribution (the number of frames) over all 9,537 videos in the training dataset of UCF101~\cite{soomro2012ucf101}. The video length is distributed between 29 and 1,776 frames, with a \emph{mean} frame count of 187 and \emph{standard deviation} of 97. Fig.~\ref{batchruntime} shows the runtime distribution over the 1,192 sampled batches in two epochs to train a 2,048-wide single-layer LSTM model on video frame features. As is standard in variable-length training, videos with similar lengths are grouped into buckets for performance. The runtime is distributed from 201 ms to 3,410 ms, with a \emph{mean} runtime of 1,235 ms and \emph{standard deviation} of 706 ms. These statistics above show that training an LSTM model for video classification exhibits inherent load imbalance.

\begin{figure}[!t]
  \centering
  \begin{subfigure}{\linewidth}
    \centering
    \includegraphics[width=.95\linewidth]{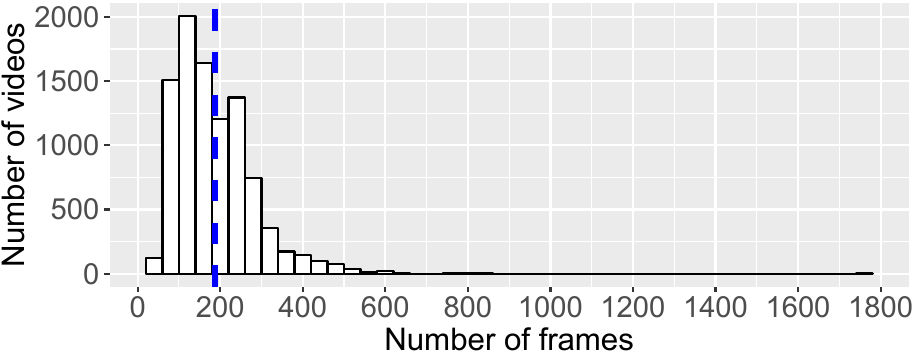}
    \caption{Video length distribution.}
    \label{videolength}
  \end{subfigure}
  \begin{subfigure}{\linewidth}
    \centering
  \includegraphics[width=.95\linewidth]{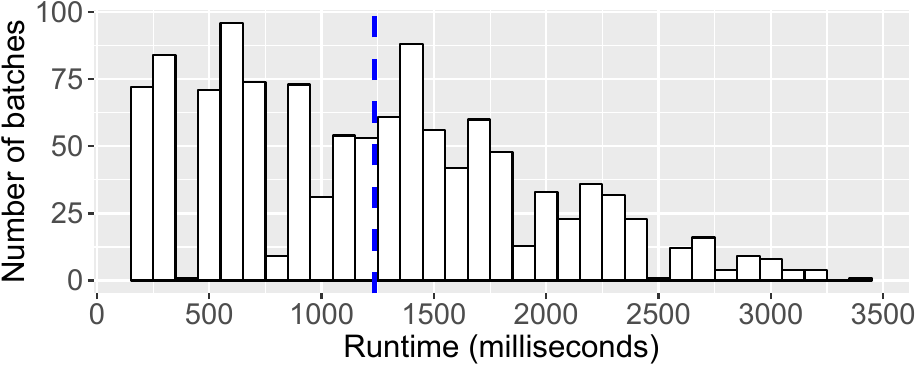}
      \caption{Runtime distribution on a P100 GPU (batch size=16).}\label{batchruntime}
  \end{subfigure}
  \caption{Load imbalance in the training of an LSTM model on UCF101~\cite{soomro2012ucf101}.}
  \label{ucf101imb}
\end{figure}

\subsection{Language Processing}

\label{imblanguagemodel}

\begin{figure}[t]
\centering\includegraphics[width=.95\linewidth]{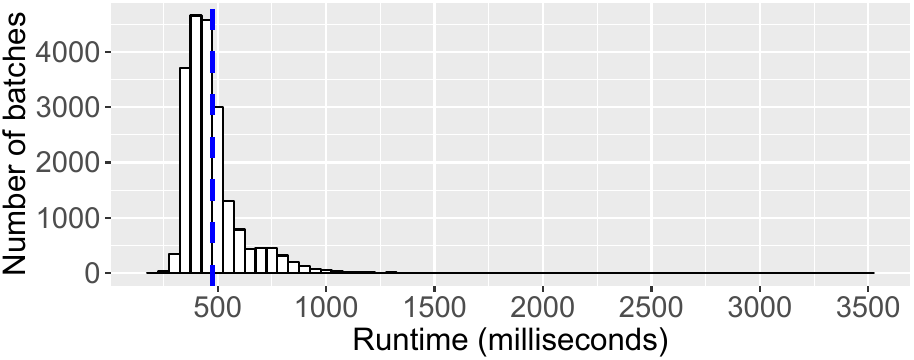}
\caption{\label{transLoad} Runtime distribution on a P100 GPU (batch size = 64), using a Transformer model on WMT16.}
\end{figure}

Transformers~\cite{vaswani2017attention} are sequence-to-sequence models that translate a sequence of words from one language to another. Different from RNN, a Transformer network replaces the recurrent structure with an attention mechanism. To train the Transformer model, the computation overhead increases with the length of the input (and output) sentences. Typically, the sentences in the training dataset for a language model have various lengths, and thus the workload is unbalanced across different batches. Fig.~\ref{transLoad} shows the runtime distribution over the 20,653 randomly sampled batches in 1/3 epoch to train a Transformer on the WMT16 dataset. The runtime is distributed from 179 ms to 3,482 ms with a \emph{mean} of 475 ms and \emph{standard deviation} of 144 ms, which shows the inherent load imbalance in language model training.

\subsection{Training in the Cloud}

\label{imbsys}

\begin{figure}[t]
\centering\includegraphics[width=.95\linewidth]{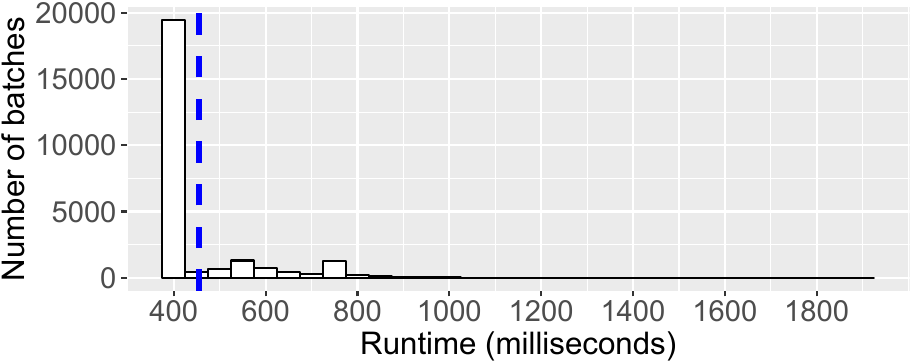}
\caption{\label{systemimb} Runtime distribution on Google Cloud with 2xV100 GPUs (batch size=256, ResNet-50 on ImageNet).}
\end{figure}

Performance variability is common in cloud computing~\cite{schad2010runtime, iosup2011performance, jackson2010performance}. Fig.~\ref{systemimb} shows the runtime distribution over the sampled batches for 5 epochs of training for the classic ResNet-50 model~\cite{he2016deep} on ImageNet~\cite{deng2009imagenet}, on a standard Google Cloud instance (\texttt{n1-standard-16} with 2x Nvidia V100 GPUs). 
The runtime is distributed from 399 ms to 1,892 ms with a \emph{mean} of 454 ms and \emph{standard deviation} of 116 ms. Since ResNet-50 on ImageNet has the same input size for different batches, the load imbalance is caused mainly by the system. 
Compared with imbalanced applications (e.g., Transformer, LSTM), the load imbalance on cloud servers is relatively light.

\section{Distributed Deep Learning}
Deep neural networks are continuously differentiable functions that are composed of multiple operators, representable by a directed acyclic graph~\cite{nature15}. The gradient of those functions can be computed using the backpropagation algorithm~\cite{lecun1998gradient}, processing the nodes in the DAG in a reverse topological order. Deep learning frameworks, such as TensorFlow~\cite{tensorflow2015-whitepaper}, typically execute parallel operations in the DAG in arbitrary order.

\begin{algorithm}[h]
  \footnotesize
  \begin{algorithmic}[1]
    \For{$t = 0$ \textbf{to} $T$}
    \State $\vec{x},\vec{y}\leftarrow$ Sample $B$ elements from dataset
    \State $w_{t}\leftarrow$ Obtain parameters from global view
    \State $\vec{z} \leftarrow \ell\left(w_{t}, \vec{x}, \vec{y}\right)$
    \State $G_{t}\leftarrow\frac{1}{B} \Sigma_{i=0}^{B}\nabla\ell\left(w_{t}, \vec{z}_i\right)$
    \State $\Delta w\leftarrow U\left(G_{t},w_{(0,\dots,t)}, t\right)$
    \State Update global view of parameters to $w_{t}+\Delta w$
    \EndFor
  \end{algorithmic}
  \caption{Distributed Minibatch SGD}
  \label{alg:sgd}
\end{algorithm}

Supervised deep neural network training typically involves first-order optimization in the form of Stochastic Gradient Descent (SGD)~\cite{sgd}. SGD optimizes the expected loss value over the ``true'' distribution of input samples by descending in the direction of a random subset of the training samples (minibatch). In a distributed data-parallel setting, the SGD algorithm (Algorithm \ref{alg:sgd}) consists of multiple learner processes, each of which updates a global view of the parameters $w$ according to a different random minibatch at the same time. Given an update rule $U$ and local minibatch of size $B$, the learners modify the global view of the parameters by using an average of the gradients $G_t$ obtained by the agents.

A straightforward manner to maintain a global view is using a Parameter Server (PS) architecture~\cite{dean12}, where one or several nodes assume the role of a PS, broadcasting up-to-date weights (line 3) to learners prior to each step and aggregating gradients from them (line 7). This enables the PS to asynchronously update the global view~\cite{hogwild}, or require a fraction of learners to send gradients before progressing to the next step~\cite{ssp13}. 

As the PS model is generally not scalable, another mode of operation implements SGD using collective operations. In such implementations, accumulating the gradients (line 7) is done via an allreduce operation, where each learner contains its own local view of the weights~\cite{dl-survey}. Horovod~\cite{sergeev2018horovod} is one such implementation over the TensorFlow framework, which also fuses several allreduce operations into one in order to reduce overhead. However, due to the arbitrary order of execution imposed by the frameworks, Horovod uses a master process for negotiation communication (achieving consensus on which parameters are sent).

A more scalable method, used in the Deep500 DSGD optimizer~\cite{deep500}, is to ensure an order of communication execution by adding control dependencies into the computation DAG, as shown in Fig.~\ref{fig:ctrl}. In the backward pass, the allreduce operations are executed in a specific order after finishing the local gradient computation. We use the same method when implementing eager-SGD.
Note that synchronizing gradient order can be avoided completely using non-blocking collectives~\cite{mpi-3.1}. In this mode, each gradient communication message is assigned to an agreed-upon numeric tag, and multiple allreduce operations may be in-flight concurrently. Prior to updating the local view of the weights, a \texttt{waitall} command must be issued. All in all, these approaches reduce overhead in imbalanced loads by overlapping communication and computation, but do not mitigate it completely.

\begin{figure}[t]
	\centering\includegraphics[width=.95\linewidth]{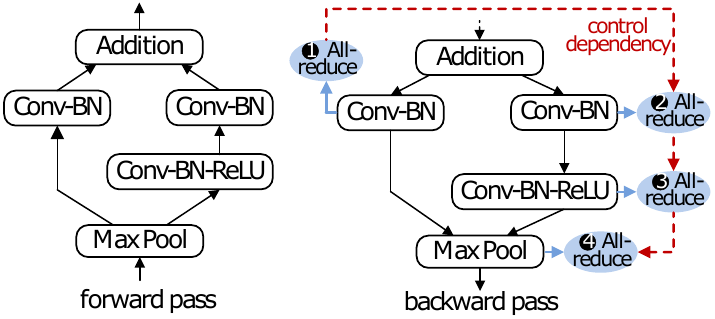}
	\caption{\label{allreduceop} Adding control dependency in the computation DAG, using a block of ResNet-50 as an example.}\label{fig:ctrl}
\end{figure}

\section{Partial Collective Operations}
\label{eagercoll}

A collective communication involves a set of processes that cooperate to
progress their internal state. Some of these operations, e.g., allreduce,
implicitly synchronize the participants: the operation cannot terminate before
the slowest process joins it. We define these collectives as
\textit{synchronous} and introduce a new class of \textit{partial} collectives
that relax the synchronization. We now discuss two variants of partial
collectives: \textit{solo} and \textit{majority}.

\subsection{Solo Collectives}

A solo collective~\cite{di2015exploiting} is a wait-free operation, which
forces the slow processes to execute the collective as soon as there is one
process executing it. This process, called \textit{initiator}, is in charge of
informing the others to join the collective.  While solo
collectives remove the synchronization delays, they change the semantics of
collective operations, which may now be completed by using stale data from the slow
processes.

\subsubsection{Schedule Activation}
\label{sec:solo_schedule_activation}

We define a \textit{schedule} as a set of operations that a process executes in
order to globally progress the collective operation. In particular, a schedule
is a directed acyclic graph (DAG) where the vertices are operations and the
edges are \texttt{happens-before} dependencies among them. We define the
following operations:
\begin{itemize}[noitemsep,topsep=0pt,parsep=0pt,partopsep=0pt,leftmargin=*]
\item Point-to-point communications: sends and receives.
\item Computations: simple computations defined between two arrays of data
items. The type of the data items is defined according to the MPI basic
types~\cite{mpi-3.1}.  
\item Non-operations (NOP): complete immediately
and are only used to build dependencies.  
\end{itemize} 
Operations can be dependent on zero, one, or more other operations (with
\textit{and} or \textit{or} logic) of the same schedule. 

The main difference between synchronous and solo collectives is the time at
which processes activate (i.e., starts executing) their schedule. For
synchronous collectives, the schedule is executed only when a process reaches
the collective function call (e.g., \texttt{MPI\_Allreduce}). We define this
activation as \textbf{internal}. 
For solo collectives, an \textbf{external} activation is also possible: the
processes start executing the schedule because of an activation message
received from the initiator, which starts broadcasting it immediately after the
internal activation of its schedule. In particular, a solo collective is
composed of two schedules: one for broadcasting the activation and the other one for
executing the collective operation.

In a solo collective, any process can become the initiator, hence any process
must be capable of broadcasting the activation message.  The activation
broadcast is implemented as a modified version of the recursive doubling
communication scheme: this is equivalent to the union of $P$ binomial trees
(optimal for small message broadcast, like the activation) rooted at the
different nodes. 

\begin{figure}
\centering\includegraphics[width=.95\linewidth]{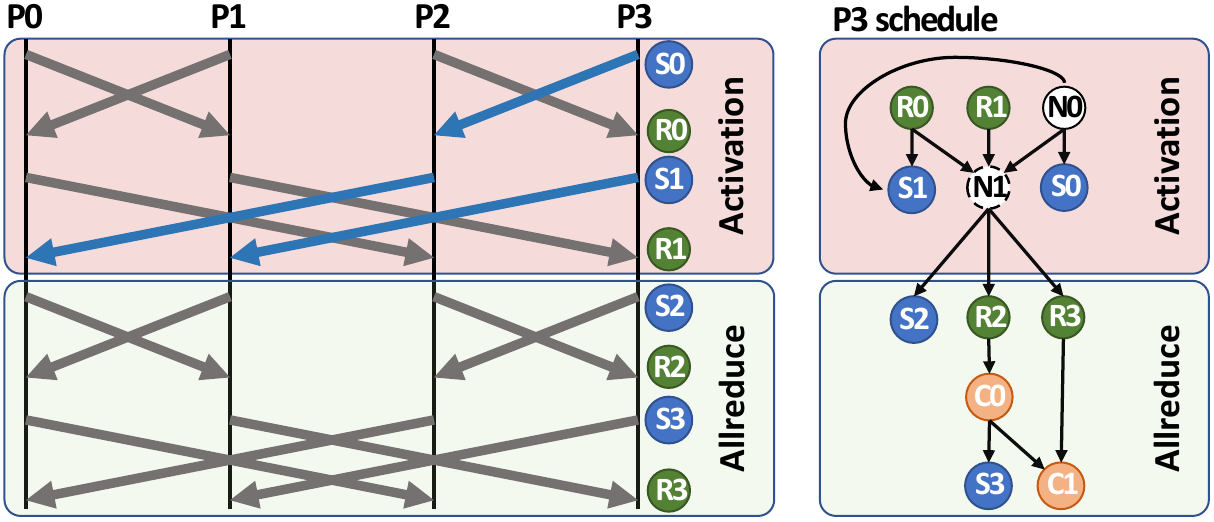}
\caption{Solo collective activation (left) and process schedule (right). 
Operations are represented by circles: blue = send, green = receive, orange =
computation, white = NOP. 
A dashed
border means the operation can be fired as soon as one of its dependencies
are satisfied.}
\label{fig:solo_activation}
\end{figure}

\paragraph{Activation example}
Fig.~\ref{fig:solo_activation} shows a solo allreduce example. On the left, we
show the global communications view that is split in two phases: activation and
allreduce.  The highlighted communication shows the activation path if the
initiator is, e.g., process P3. For the allreduce, we use a recursive doubling
implementation. Note that any collective implementation that can be expressed
as a schedule can be linked to the activation phase.
On the right we show the internal schedule of process P3. An internal
activation (i.e., P3 making the function call explicitly) translates in the
execution of NOP 0 (N0): this leads to the send operations S0 and S1 being
fired to start broadcasting the activation message and to the execution of N1,
which signals the activation of the allreduce schedule. Alternatively, if P3 is
not the initiator, it will receive a message in receive R0 or R1: if the
activation is received by R0, then P3 has to forward the activation message to
P1 with send S1 (i.e., P3 is an internal node of the activation binomial tree).
Also in this case NOP N1 will be executed, leading to the execution of the
allreduce schedule.

\paragraph{Multiple initiators}
Multiple processes may join the collective at the same time: in this case we
need to ensure that the collective is executed only once.  To address this
issue, we set the operations to be \textit{consumable}, meaning that the same
operation cannot be executed twice.  
For example, let us assume that nodes P2 and P3 reach their internal activation
at the same time. When P3 receives the activation message from P2 (i.e.,
through R0) there are two possible cases: 1) S1 is still not consumed and then
it is executed; 2) S1 has been fired due to the internal activation and will
not be executed a second time. NOPs are also consumable, hence N1 (i.e., the
activation) can be executed only once. 

\paragraph{Persistent schedules}

Processes can be asked to join a solo collective only once before they
reach their internal activation: once the schedule is executed, it needs
to be re-created by the application in order to be executed again.
To enable multiple asynchronous executions of solo collectives, we introduce
\textit{persistent schedules}. Such schedules transparently replicate themselves once
executed, able to serve a new solo collective without requiring 
application intervention. Multiple executions of the same solo collective
overwrite the data in the receive buffer, which always contains the value
of the latest execution.

\subsection{Majority Collectives}

An issue of solo collectives is that if one or few processes are always faster
than the others, then the collective will always complete by taking the stale
data of the slower processes. In cases like DNN training, this scenario may
negatively impact the convergence because the training will advance only
considering the updates of few processes. To overcome this issue, we introduce
\textit{majority} collectives, which requires at least half of the processes to join before
completing. We implement majority collectives by not letting any process become
the initiator, as in solo collectives. Instead, at each execution of a
persistent schedule, the processes designate an initiator by randomly selecting
a rank (consensus is achieved by using the same seed for all the processes).
When a process joins the collective (i.e., internal activation), it checks
whether it is the designated initiator: only in that case it keeps running the
internal activation followed by the actual collective schedule.

 
We now discuss how the above described implementation can provide a statistical
guarantee that at least half of the processes on
average contribute to the collective. Suppose the same collective operation is called by many iterations,
such as in model training. We sort all the $P$ processes by the
time they reach a collective operation. Since the probability that any
process is specified as the initiator is equal to 1/$P$, the expectation of the
randomly specified initiator is the $P$/2-th process among the sorted
processes, namely on average half of the processes reach the collective
operation earlier than the initiator. For a workload distribution
with one mode and a tail, such as in Figs.~\ref{ucf101imb},~\ref{transLoad}, and~\ref{systemimb}, the probability that part of the processes reach the collective at a similar
time to the initiator is high; then, more than half of the processes on average actively participate in the operation.

\subsection{Asynchronous Execution by Library Offloading}

The schedule of a partial collective can be asynchronously
executed with respect to the application. We develop \textit{fflib2}, a communication library that allows to express communication schedules and offload their execution to the library itself. %
The schedule execution can take place on the application thread (i.e., when the
application enters the library), or on an auxiliary thread. Once the application
creates and commits a schedule, the library starts executing all the operations
that have no dependencies. The remaining ones are executed as their dependencies
are satisfied. 

\subsection{Discussion}
Offloading the schedule execution to the network interface card (NIC) can
provide different advantages such as asynchronous execution, lower latency, and
streaming processing. Di Girolamo et al.~\cite{di2015exploiting} show how solo
collectives can be offloaded to Portals 4~\cite{barrett2018portals} NICs by
using triggered operations. This approach is limited by the amount of NIC
resources that bounds the number of times a persistent schedule can be executed
without application intervention. 
This limit can be removed by implementing the schedule execution with the sPIN
programming model~\cite{hoefler2017spin}, which allows to execute user-defined
code on the NIC. A sPIN implementation of \textit{fflib2} would then be able to replicate
the schedule on-the-fly upon completion.

\section{Eager-SGD algorithm}
\label{eagerSGDalgorithm}

Algorithm~\ref{alg:esgd} illustrates the main procedure of eager-SGD. Instead of calling a synchronous allreduce in the distributed optimizer (Fig.~\ref{fig:ctrl}) to accumulate the gradients, eager-SGD uses the partial allreduce operations (Line 7). Either solo or majority allreduce can be used depending on the severity of load imbalance. 

Fig.~\ref{eagersgd} presents an example of how eager-SGD works with partial collectives, in which w$^p_t$ and G$^p_t$ represent the weights and the gradients calculated on process $p$ at training step $t$, respectively, and $U\left(G, w\right)$ represents the update rule. In step $t$, suppose process P1 is faster than process P0. P1 finishes the computation of G$^1_t$ and then triggers the partial allreduce operation. Since P0 does not finish the computation of G$^0_t$ at this time, it only passively contributes null gradients G$_{null}$ to the partial allreduce at step $t$. After P0 finishes the computation of G$^0_t$, it finds out that the partial allreduce at step $t$ is already finished by checking the results in the receive buffer. P0 updates the weights of step $t+1$ using G$^1_t$ stored in the receive buffer of the partial allreduce and G$^0_t$ becomes the stale gradients. The stale gradient G$^0_t$ is then stored in the send buffer. If P0 does not catch up with P1 at step $t+1$, P0 will passively participate in the partial allreduce again and contribute G$^0_{t}$. If P0 catches up with P1 at step $t+1$ (as in the case shown in Fig.~\ref{eagersgd}), P0 will add G$^0_{t}$ and G$^0_{t+1}$ (calculated in step $t+1$) together, and contribute the accumulated gradients G$^{0\prime}_{t+1}$ to the partial allreduce; P0 resets the send buffer to G$_{null}$ after finishing allreduce. 

\begin{algorithm}[t]
  \footnotesize
  \begin{algorithmic}[1]
    \State $b$ is local batchsize for $P$ processes
    \For{$t = 0$ \textbf{to} $T$}
    \State $\vec{x},\vec{y}\leftarrow$ Each process samples $b$ elements from dataset
    \State $\vec{z} \leftarrow \ell\left(w_{t}, \vec{x}, \vec{y}\right)$
    \State $G_t^{local}\leftarrow\frac{1}{b} \Sigma_{i=0}^{b}\nabla\ell\left(w_{t}, \vec{z}_i\right)$
    \State $G_t^{global}\leftarrow\frac{1}{P}\ partial\_allreduce\left(G_t^{local}\right)$
    \State $\Delta w\leftarrow U\left(G_t^{global},w_{(0,\dots,t)}, t\right)$
    \State $w_{t+1} \leftarrow w_{t}+\Delta w$
    \EndFor
  \end{algorithmic}
  \caption{Eager-SGD}
  \label{alg:esgd}
\end{algorithm}

\begin{figure}[t]
\centering\includegraphics[width=.95\linewidth]{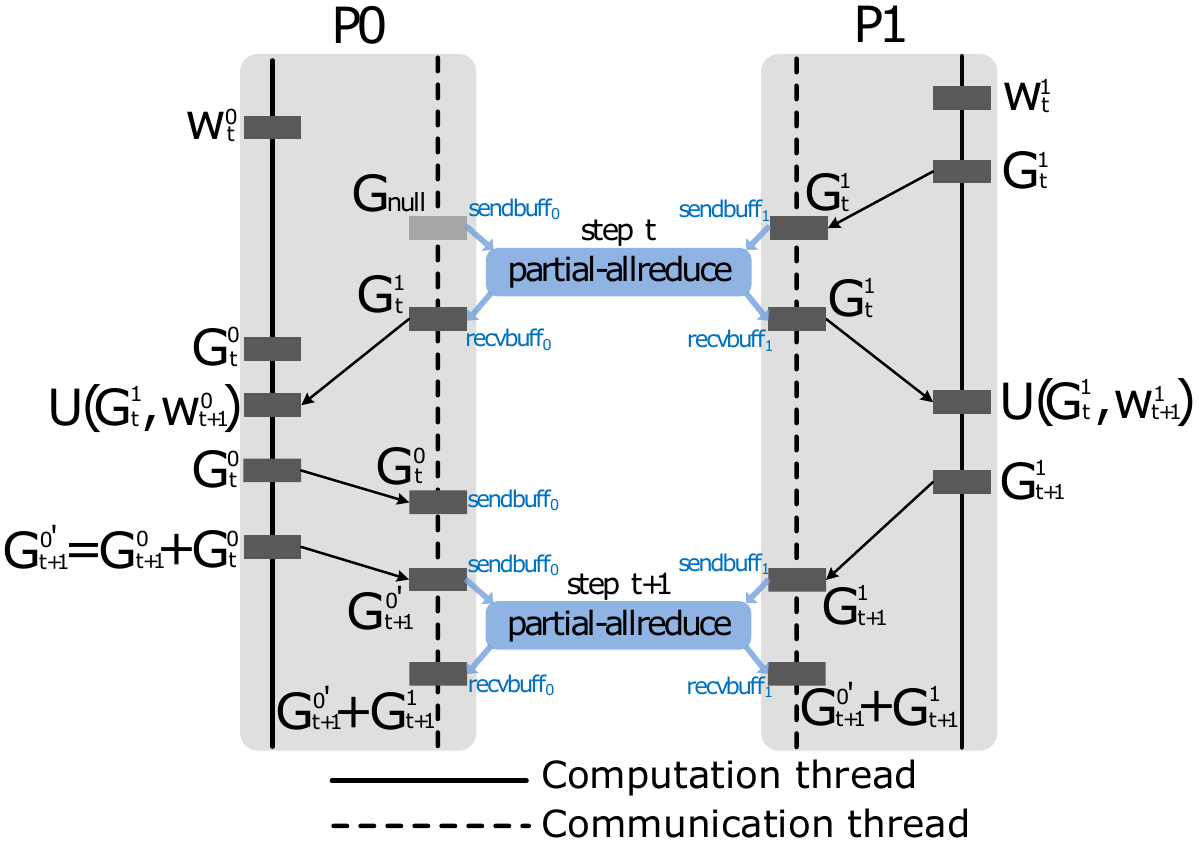}
\caption{\label{eagersgd} Partial collective operations in eager-SGD.}
\end{figure}

In severe load imbalance situations, some slower processes may lag behind by more than one step. The data in the receive buffer of the partial allreduce will then be overwritten and only the latest data in the receive buffer can be seen, which results in different weights on different processes. This may result in slightly lower accuracy as shown in Section~\ref{resnetevaluation}. Thus, we periodically synchronize the models across all processes to eliminate the side effect. Since we only synchronize the models every few epochs, the overhead can be ignored.

\section{Correctness and Convergence Guarantees}

\subsection{System Model}

\setlist[itemize]{leftmargin=*}

We prove that, under a reasonable set of modeling assumptions, the eager-SGD algorithm will still converge. 
We assume a system with $P$ asynchronous processors indexed as $i \in \{0, 1, \ldots, P - 1\}$,  which take steps at different speeds.

For simplicity, we break down the execution at each processor into \emph{steps}: at step $t$, we assume that each processor $i$ has collected a local view of the parameters, which we denote by $w_t^i$.
We then proceed as follows: 
the processor computes the gradient $G_t^i$ on a randomly sampled mini-batch, with respect to the local view $w_t^i$, and enters an \texttt{partial-allreduce} for the step, whose goal is to attempt to communicate its current parameter updates to other processors. At the end of this, the process obtains its next view of the parameters $w_{t+1}^i$, which it will use in the following step $t + 1$. 

From a global perspective, we can split the execution in serial fashion into \emph{rounds}, where each round can be mapped to the \texttt{partial-allreduce} of corresponding index. Without loss of generality, we assume that each processor participates in each round $t$, since it eventually submits an update to the corresponding \texttt{partial-allreduce}, which we denote by $\texttt{ADS}(t)$, for \emph{asynchronous distributed sum}. However, its update may or may not be delivered to the other processors. Each \texttt{partial-allreduce} has the following semantics: 

\begin{itemize}
	\item \textbf{(Invocation)} Each process $i$ proposes a $d$-dimensional vector $R^i_t$, corresponding to its current proposed update, to $\texttt{ADS}(t)$.
	\item \textbf{(Response)} Each process $i$ receives a tuple $\langle U_t, s^i_t \rangle$, where $U_t$ is the $d$-dimensional update to the parameter set corresponding to round $t$, as decided by the shared object $\texttt{ADS}(t)$, and  $s^i_t$ is a boolean stating whether the update by process $i$ has been included in $U_t$.
\end{itemize}

We can therefore rephrase the algorithm as having each process invoke the $\texttt{ADS}(t)$ object in each round, with its current update. 
If its update is not ``accepted'' ($s^i_t = \texttt{false}$) then the processor simply adds it to its update in the next iteration. 
The ADS objects we implement provide the following guarantees. 

\begin{lemma}
\label{lem:ads} 
	Each \emph{ADS} object ensures the following: 
	\begin{enumerate} 
	\item \textbf{(Liveness)} The $\texttt{ADS}(t)$ object eventually returns an output at every invoking process. 
	\item \textbf{(Safety)} The output is consistent, in the sense that (1) it is a correct \emph{average} of a subset of the proposed updates in the round; (2) the returned bits reflect its composition; and (3) the output is \emph{the same} at every invoking process. 
	\item \textbf{(Quorum Size)} The subset of proposed updates included in the output is of size $Q \geq 1$, where $Q$ is a lower bound parameter ensured by the algorithm.
	\item \textbf{(Staleness Bound)} There exists a bounded parameter $\tau$ such that any update by a process can be rejected by the $\texttt{ADS}$ objects for at most $\tau$ consecutive rounds from the time it was generated before being accepted. 
	\end{enumerate}
\end{lemma}

\begin{proof}
  The proof of the above properties follows directly from the structure of the \texttt{partial-allreduce} algorithm, and is therefore skipped. 
\end{proof}

\subsection{Convergence Proof}

We now show that these properties are sufficient for eager-SGD to ensure convergence for a standard class of smooth non-convex objectives. In the following, all norms are $\ell_2$-norms, unless otherwise stated. 

\begin{assumption}[Loss Function]
\label{assumption:loss}
  We assume that our objective loss function $f: \mathbb{R}^d \rightarrow \mathbb{R}$ satisfies the following standard properties:
  \begin{itemize}
    \item  \textbf{(Lower Bound)} The function $f$ is bounded from below, that is, there exists a finite value $m$ such that, $\forall \vec{x} \in  \mathbb{R}^d, f(\vec{x}) \geq m$. 
    
    \item  \textbf{(Smoothness)} The function $f$ is $L$-smooth, i.e. $$\forall~\vec{x},\vec{y} \in \mathbb{R}^d,~\|\nabla f\left(\vec{x}\right) - \nabla f\left(\vec{y}\right)\| \leq L\|\vec{x}-\vec{y}\|~\text{for}~L>0.$$
  \end{itemize}
\end{assumption}

Further, we make the following standard assumptions about the gradients generated by the nodes:
\begin{assumption}[Gradients]
\label{assumption:grad}
  For any round $t$ and processor $i$, the gradients $G^i_t$ generated by the processes satisfy the following, where expectations are taken with respect to the random data sampling at round $t$.

  \begin{itemize}
    \item  \textbf{(Unbiasedness)} The stochastic gradients are unbiased estimators of the true gradients:
    $$\forall \vec{x} \in \R^d,~\E \left[G_t^i(\vec{x})\right] = \nabla f(\vec{x}),$$     
    \item \textbf{(Second Moment Bound)} There exists a constant $M$ such that 
    $$\forall \vec{x} \in \R^d,~\E\left[\|G_t^i \left(\vec{x}\right)\|^2 \right] \leq M^2.$$
    
  \end{itemize}
\end{assumption}

\paragraph{Analytic View of the Algorithm.}
Let us fix a global round $t + 1$, and consider the view of an arbitrary process $i$, $w_{t + 1}^i$ at the beginning of this round.  Recall that this view consists of the view returned by the object $\texttt{ADS}(t)$. 
Therefore, by Lemma~\ref{lem:ads}, this view must include the sum of at least $Q$ distinct gradients generated in each previous round, possibly together with some additional gradients, some of which are included in their corresponding round, and some of which are delayed. Conversely, if we consider the gradients which have been proposed to \texttt{ADS} objects by all nodes by time $t$ and are \emph{not included} in this view, we have that there can be at most $P - Q$ such gradients for any previous round, up to maximum time $\tau$ in the past. We formalize this observation as follows. 

Define recursively the auxiliary random variable $\Lambda_t$ such that 
for every round $t \geq 0$,
$$\Lambda_{t +1} = \Lambda_t - \frac{\alpha}{P} \sum_{i = 0}^{P - 1} G^i_t( w^i_t ),$$ where $\alpha > 0$ is the learning rate, which we assume to be constant. Without loss of generality, we set $\Lambda_0 = 0^d$.  Intuitively, $\Lambda_t$ would like to follow the ``clean'' SGD iteration, by including all the gradients generated by the end of round $t$. 
However, one technical issue is that these gradients are generated not with respect to the model $\Lambda_{t - 1}$ (which would allow us to perform a standard SGD analysis) but with respect to the partial views $w^i_t$. We will overcome this obstacle by leveraging the fact that the partial view $w^i_t$ cannot be too far from $\Lambda_t$. 
More precisely, the discussion in the previous paragraph implies:

\begin{lemma}
 \label{lem:view}
  For any $t \geq 0$ and process $i$, we have: 
  $$\E [ \| \Lambda_t - w^i_t \|^2]  \leq \alpha^2 \tau M^2 (P - Q ) / P^2.$$
\end{lemma}
\begin{proof}
  Let $\delta^j_t$ be a binary indicator random variable that is true if the gradient generated by process $j$ at iteration $t$ is \emph{not delivered} by the $\texttt{ADS}(t)$ object. Then, we have that:
  \begin{eqnarray} 
    \| \Lambda_t - w^i_t \|^2 = \| \sum_{t = 1}^\infty \alpha \sum_{j = 1}^P \delta^j_t G^j_t / P \|^2 \\ = 
    \| \sum_{t = 1}^\tau \alpha \sum_{j = 1}^P \delta^j_t G^j_t / P \|^2 \\ \leq \sum_{t = 1}^\tau (\alpha^2 / P^2) \sum_{j = 1}^P \delta^j_t \| G^j_t \|^2,
  \end{eqnarray}
  where we have used the properties stated in Lemma~\ref{lem:ads} (in particular the Staleness Bound), and the triangle inequality. 
  Next, we notice that (1) the expected squared norm of each of the missing gradients is bounded by $M^2$ (by the second moment bound), and that 
  (2) there can be at most $P - Q$ delayed gradients from each round (by the Quorum Size bound). 
  This finally implies the claimed inequality:
  \begin{eqnarray} 
  \E [ \| \Lambda_t - w^i_t \|^2 ] \leq & (\alpha^2 / P^2) \sum_{t = 1}^\tau \sum_{j = 1}^P \delta^j_t \E [ \| G^i_t \|^2 ] \\ \leq & \alpha^2 \tau M^2 (P - Q ) / P^2.
  \end{eqnarray}
  
\end{proof}

\paragraph{Convergence Bound.} 
Finally, we put all the machinery together to obtain the following convergence bound:

\begin{theorem}[Eager-SGD Convergence]
\label{thm:convergence}
  Consider an arbitrary objective function $f$ and gradient sampling scheme satisfying Assumptions~\ref{assumption:loss} and~\ref{assumption:grad}. 
  Fix the success parameter $\epsilon > 0$. 
  Then, if we execute the eager-SGD algorithm for constant learning rate value 
  $$\alpha \leq \min\left( \frac{\sqrt{ \epsilon } P } {\sqrt {12 L^2 \tau M^2 (P - Q)}}, \frac{\sqrt{\epsilon} P }{\sqrt{4L \tau M^2 (P - Q)} } , \frac{\epsilon}{ 12 M^2 L} \right)$$ for  $T = \Theta \left( \frac{ f ( w_0 ) - m}{\epsilon \alpha} \right)$ iterations, we are guaranteed to reach some iterate $w_{t^\star}$ with $1 \leq t \leq T$ such that $$\E \| \nabla f(w_{t^\star}) \|^2 \leq \epsilon.$$ 
\end{theorem}

\begin{proof}
We begin from the definition of $\Lambda_t$: 
\begin{equation}
  \Lambda_{t +1} = \Lambda_t - \frac{\alpha}{P} \sum_{i = 0}^{P - 1} G^i_t( w_t ).
\end{equation} 

We will first prove the above statement for the iterate $\Lambda_t$, and then will extend the proof for $w_t$. 
For simplicity, let us denote $G_t = \sum_{i = 0}^{P - 1} G^i_t( w_t ).$  
We can use the Taylor expansion of $f(\Lambda_{t + 1})$ around $\Lambda_t$ and the smoothness condition to obtain the following inequality:

  \begin{align*}  
    f(\Lambda_{t+1})& \le f\left(\Lambda_t\right) + (\Lambda_{t+1}-\Lambda_{t})^T\nabla f\left(\Lambda_{t}\right) + \frac{L}{2}\|\Lambda_{t+1}-\Lambda_{t}\|^2\\    
    & = f(\Lambda_t)-{\alpha} \nabla f(\Lambda_t)^T\nabla f(\Lambda_t)+\frac{{\alpha}^2 L}{2P^2} \| G_t\|^2 + \\ 
    & + {\alpha} (\nabla f(\Lambda_t) -  G_t / P)^T\nabla f(\Lambda_t).  
  \end{align*}

We can therefore apply the expectation with respect to the random sampling at step $t$, the second moment bound assumption: 
\begin{align*}
  \E\left[f(\Lambda_{t+1})| \Lambda_t \right] \le & f(\Lambda_t)-{\alpha} \| \nabla f(\Lambda_t) \|^2 
  + \frac{\alpha^2 L}{2} M^2  \\ + & \alpha ( \nabla f(\Lambda_t) - \nabla f(w_t) )^T \nabla f(\Lambda_t ).
\end{align*}

To bound the last term, we can now apply the Cauchy-Schwarz inequality and the fact that the gradients are $L$-Lipschitz: 
\begin{align*}
  \E\left[f(\Lambda_{t+1})| \Lambda_t \right] \le & f(\Lambda_t)-{\alpha} \| \nabla f(\Lambda_t) \|^2 
  + \frac{\alpha^2 L}{2} M^2  \\ + & \alpha L \| \Lambda_t - w_t \| \|\nabla f(\Lambda_t ) \|.
\end{align*}

To further bound the last term, we can apply the classic inequality $a^2 + b^2 \geq 2ab$ together with Lemma~\ref{lem:view} to obtain:
\begin{align*}
  \E\left[f(\Lambda_{t+1})| \Lambda_t \right] \le f(\Lambda_t)-{\alpha} \| \nabla f(\Lambda_t) \|^2 
  + \frac{\alpha^2 L}{2} M^2  \\ +  \alpha  \|\nabla f(\Lambda_t ) \|^2 / 2 + \frac{\alpha^3  L^2 \tau M^2 (P - Q) } {2P^2}.
\end{align*}

\noindent Rearranging terms and taking total expectation:  
\begin{align*}
  \E\left[ \| \nabla f(\Lambda_t) \|^2 \right] \leq & \frac{ 2 \E\left[f(\Lambda_{t}) - f(\Lambda_{t + 1}) \right] }{\alpha} + {\alpha M^2 L}
   \\  + & \alpha^2 \tau L^2  M^2 (P - Q) / P^2.
\end{align*}

Summing across all $t$ and dividing by $T$, we get:
\begin{align*}
  \min_{1 \leq t \leq T} E\left[ \| \nabla f(\Lambda_t) \|^2 \right] \leq  \frac{1}{T} \sum_t \E\left[ \| \nabla f(\Lambda_t) \|^2 \right]  \leq  \\ 
  \leq  \frac{ 2 \left(f(\Lambda_0) - m \right)  }{\alpha T} + \alpha M^2 L + \alpha^2 L^2 \tau M^2 (P - Q) / P^2.
\end{align*}

We now study the set of conditions for each of the three RHS terms to be less than $\epsilon / 12$.
We have that it is sufficient for the following three conditions to hold: 

\begin{enumerate}
  \item $T \geq \frac{24 (f ( \Lambda_0 ) - m )}{\alpha \epsilon};$
  \item $\alpha \leq \frac{\epsilon}{12 M^2 L};$
  \item $\alpha \leq \frac{\sqrt{ \epsilon } P } {\sqrt {12 L^2 \tau M^2 (P - Q)}}.$
\end{enumerate}

All these conditions hold by assumption from the theorem statement.
We have therefore obtained that there exists $t^\star$ such that $\| \nabla f( \Lambda_{t^\star} ) \|^2 \leq \epsilon / 4$. 
However, by smoothness and Lemma~\ref{lem:view} we know that 
$$ \E \| \nabla f( \Lambda_{t^\star} ) - \nabla f( w_{t^\star} )  \|^2 \leq \alpha^2 L \tau M^2 (P - Q ) / P^2 \leq \epsilon / 4,$$

\noindent where we have used the assumption in the theorem statement on the upper bound on $\alpha$. 
Finally, we can apply the classic inequality $\| a + b \|^2 \leq 2( \| a \|^2 + \| b \|^2 )$ to obtain that 
$$\E \| \nabla f( w_{t^\star} ) \|^2 \leq \epsilon.$$

\end{proof}

\paragraph{Discussion} We make the following observations regarding the bound. 
First, we note that, since we analyze non-convex objectives, we must settle for a weaker notion of convergence than in the convex case (where we can prove convergence to a global minimum): specifically, we prove that, for a given sequence of learning rates, the algorithm will converge to a point of negligible gradient. 
Second, we note the dependence in $\sqrt{\tau}$ and $\sqrt{(P - Q)}$ for the number of iterations to convergence, i.e.: 
$$T \geq \Theta\left( \frac{(f(w_0) - m ) \sqrt{\tau (P - Q)}} {P \epsilon^{3/2}} \right).$$ 

Thus, we would like the maximum delay and the number of ``missed'' gradients per round to be minimized. However, obviously, having no stragglers would imply higher synchronization cost. 
This suggests that, in practice, the algorithm should trade off the additional performance cost of synchronization with the slower convergence due to delayed gradient information.

\begin{table*}[!t]
	\caption{Neural networks used for evaluation.}
	\label{networks}
	\centering
	\small
	\begin{tabular}{l l r r r r r r}
		\toprule
		Tasks & Models & Parameters & Train data size & Batch size & Epochs & Processes \tabularnewline
		\midrule
		Hyperplane regression & One-layer MLP & 8,193 & 32,768 points & 2,048 & 48 & 8\tabularnewline
		Cifar-10 & ResNet-32~\cite{he2016deep} & 467,194 & 50,000 images & 512 & 190 & 8 \tabularnewline
		ImageNet~\cite{deng2009imagenet} & ResNet-50~\cite{he2016deep} & 25,559,081 & 1,281,167 images & 8,192 & 90 & 64 \tabularnewline
		UCF101~\cite{soomro2012ucf101} & Inception+LSTM~\cite{yue2015beyond} & 34,663,525 & 9,537 videos & 128 & 50 & 8 \tabularnewline
		\bottomrule
	\end{tabular}
\end{table*}

\section{Evaluation}
\label{eval}

Experiments are conducted on the CSCS Piz Daint supercomputer with Cray Aries interconnect. Each XC50 compute node contains a 12-core Intel Xeon E5-2690 CPU with 64 GiB RAM, and one NVIDIA Tesla P100 GPU. The communication library is Cray MPICH 7.7.2. We use one MPI process per node and utilize the GPU for acceleration in all following experiments. First, we evaluate the performance of the partial collective operations using a microbenchmark. Then, we use the different neural networks summarized in Table~\ref{networks} to compare our eager-SGD with the state-of-the-art synch-SGD implementations (Horovod~\cite{sergeev2018horovod} and Deep500~\cite{deep500}), the asynchronous centralized SGD~\cite{tensorflow2015-whitepaper}, and the gossip-based SGDs~\cite{syncring, assran2018stochastic}, under simulated and real workload imbalance environments.

\begin{figure}[ht!]
\centering\includegraphics[width=.92\linewidth]{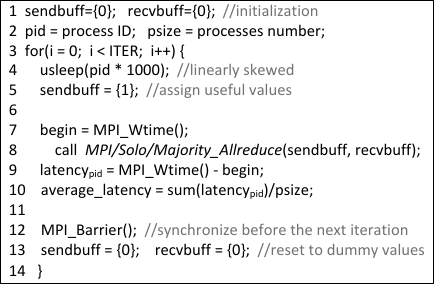}
\caption{\label{benchmark} Microbenchmark used to test the latency of the collective operations.}
\end{figure}

\subsection{Partial Allreduce Operations}

We design a microbenchmark, shown in Fig.~\ref{benchmark}, to evaluate the performance of partial allreduce operations (\textit{fflib2}) and MPI\_Allreduce (Cray MPICH) with unbalanced workload. All the processes are linearly skewed before calling the collective operations and the average latency among all the processes is recorded. The microbenchmark is a special case with severe load imbalance, which is useful to verify the statistical guarantee of majority allreduce. Experimental results on 32 processes are presented in Fig.~\ref{eagerallreduce}. Compared with MPI\_Allreduce, solo and majority allreduce operations reduce the latency by on average 53.32x and 2.46x, respectively. This is because all the processes (except the slowest one) for MPI\_Allreduce are delayed; solo allreduce is not delayed since the fastest process will trigger the operation immediately; and majority allreduce has to wait for a randomly specified process to trigger the operation, and thus it is moderately delayed.

\begin{figure}[ht!]
\centering\includegraphics[width=.92\linewidth]{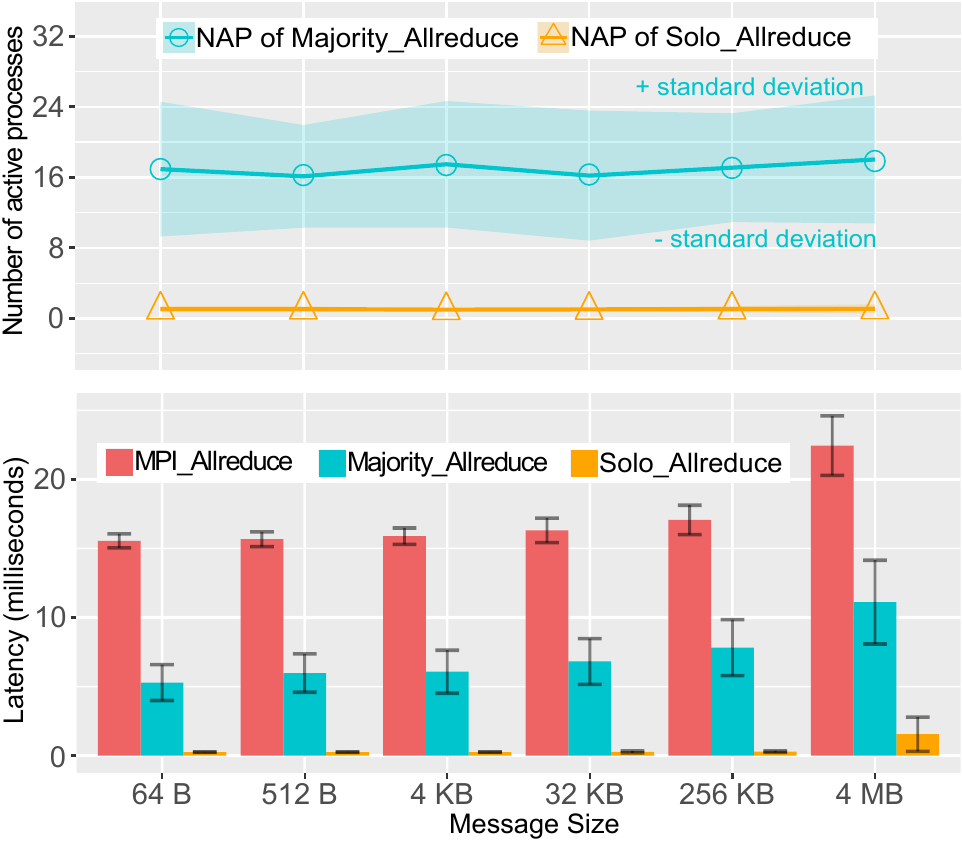}
\caption{\label{eagerallreduce} Average latency comparison between MPI\_Allreduce and partial allreduce running on 32 processes by 64 iterations. Processes are linearly skewed by injecting load imbalance from 1 ms to 32 ms.}
\end{figure}

For the partial collective operations, we refer to the initiator together with the processes that arrive at the operation before the initiator as the \textit{active processes}, which contribute the latest data (line 5 in Fig.~\ref{benchmark}). 
The other processes only contribute null values (line 13 in Fig.~\ref{benchmark}). 
For solo allreduce, since the fastest process is the initiator and all the processes are fully skewed, the Number of Active Processes (\textit{NAP}) is around 1, as shown in Fig.~\ref{eagerallreduce}. For majority allreduce, since the initiator is randomly specified, the expectation of \textit{NAP} is half of the total processes. On average 16 out of 32 processes for majority allreduce are active processes, which means half of the processes contribute the latest data when the processes are fully skewed.

\begin{figure}[!h]
  \centering
  \begin{subfigure}{\linewidth}
    \centering\includegraphics[width=.91\linewidth]{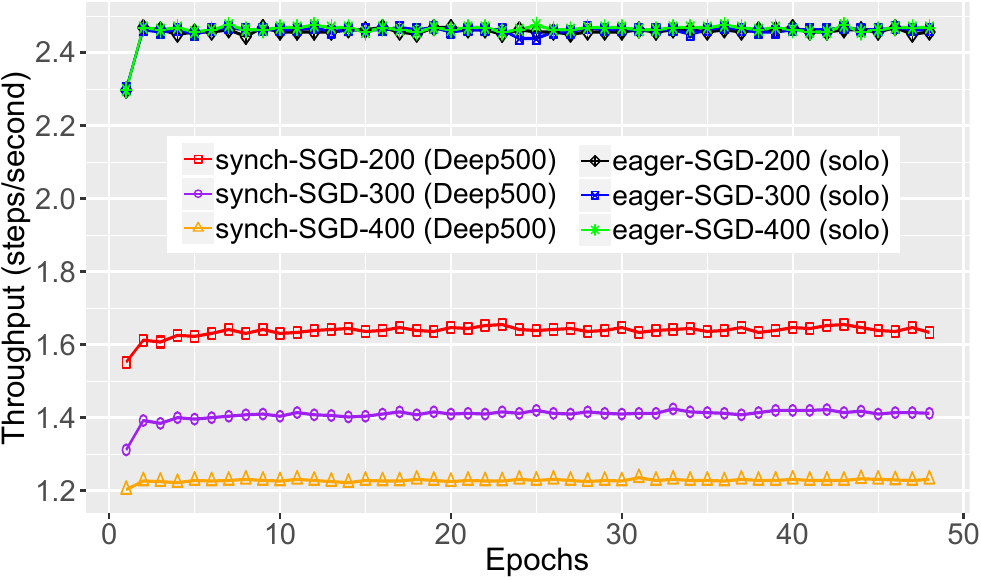}
	\caption{Throughput comparison.}\label{hpregression-tp}
  \end{subfigure}
  \begin{subfigure}{\linewidth}
	\centering\includegraphics[width=.93\linewidth]{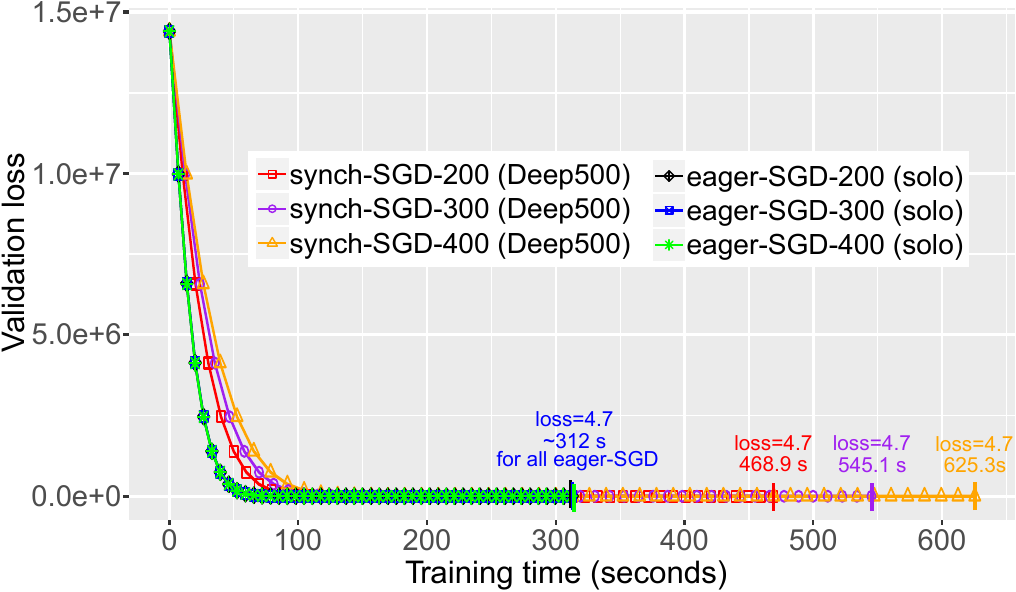}
	\caption{Validation Loss Comparison.}\label{hpregression-loss}
  \end{subfigure}

  \caption{Comparison between synch-SGD and eager-SGD for hyperplane regression using 8 processes. "synch/eager-SGD-200/300/400" represent 200/300/400 ms load imbalance injection, respectively. Each point is at the boundary of one epoch.}
  \label{hpregression}
\end{figure}

\begin{figure}[!h]
  \centering
  \begin{subfigure}{\linewidth}
    \centering\includegraphics[width=.92\linewidth]{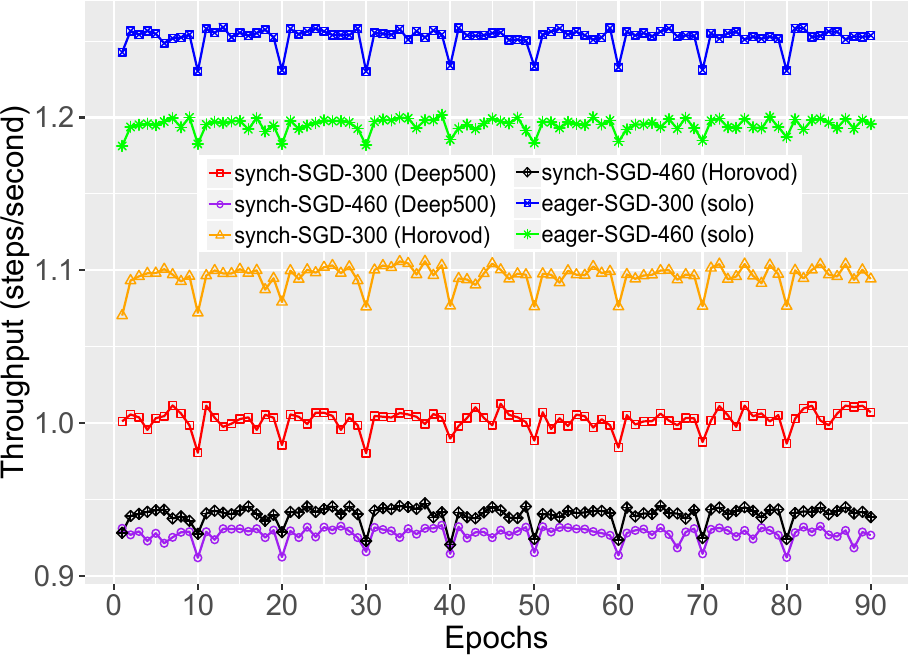}
	\caption{Throughput comparison. Each point is at the boundary of one epoch.}\label{tp-imgnet}
  \end{subfigure}
  \begin{subfigure}{\linewidth}
	\centering\includegraphics[width=.92\linewidth]{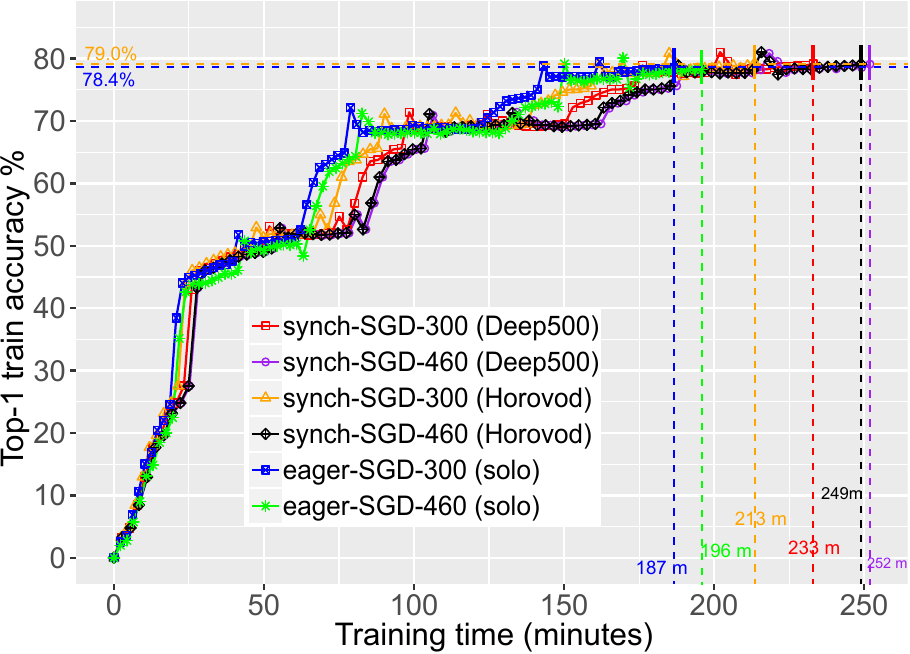}
	\caption{Top-1 training accuracy. Each point is at the boundary of one epoch.}\label{traintop1-imgnet}
  \end{subfigure}
  \begin{subfigure}{\linewidth}
    \centering\includegraphics[width=.92\linewidth]{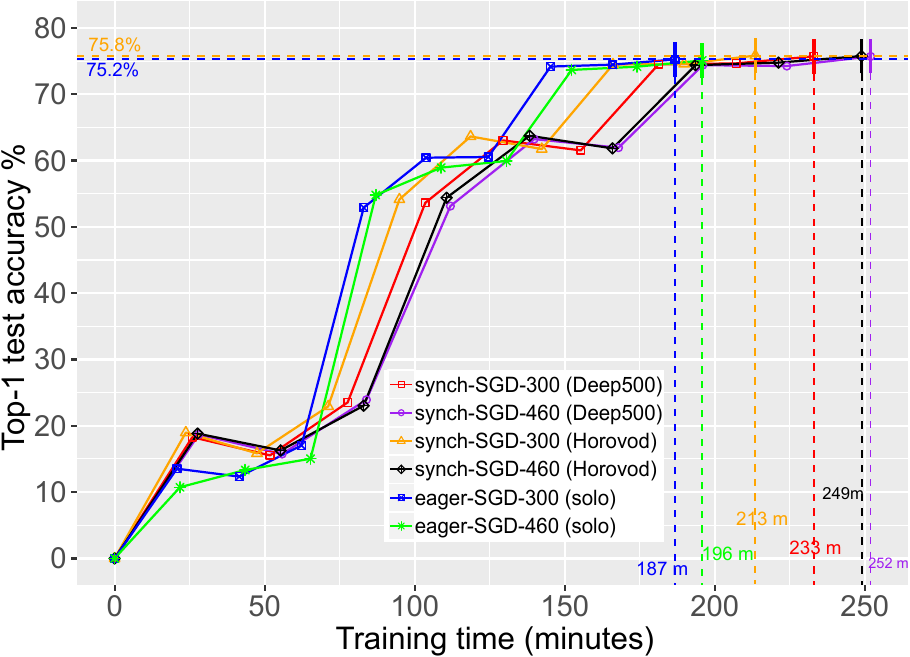}
    \caption{Top-1 test accuracy. Each point is at the boundary of every 10 epochs.}\label{testtop1-imgnet}
  \end{subfigure}

  \caption{Comparisons between synch-SGD and eager-SGD for ResNet-50 on ImageNet using 64 processes. "synch/eager-SGD-300/460" represent 300/460 ms load imbalance injection, respectively.} 
  \label{acc-imagenet}
\end{figure}

\subsection{Throughput and Convergence with Simulated Workload Imbalance}

We use three networks shown in Table~\ref{networks}, including a multilayer perceptron (MLP), ResNet-32, and ResNet-50, to evaluate the performance of eager-SGD with simulated workload imbalance. From the application perspective, these three networks have balanced workload during the distributed training, since each batch has equivalent workload. We manually inject delays to simulate the load imbalance environment caused by the training system, as discussed in Section~\ref{imbsys}.

\subsubsection{Hyperplane Regression, Light Load Imbalance} 
We generate both training and validation datasets for a 8,192-dimensional hyperplane regression using the equation:
$y=a_0x_0+a_1x_1+...+a_{8191}x_{8191}+noise$, where ($x_0, x_1,..., x_{8191}$) is the input vector and $y$ is the label. An one-layer MLP is used to learn the coefficients ($a_0, a_1,..., a_{8191}$) of the hyperplane. We use 8 processes with the total batch size of 2,048 to train the model for 48 epochs. To simulate the load imbalance environment, we randomly select one process out of the 8 processes at every training step to inject a certain amount of delay, according to the variability shown in Fig.~\ref{systemimb}. 

The throughput comparison between synch-SGD (Deep500) and eager-SGD (using solo allreduce) is shown in Fig.~\ref{hpregression-tp}. With 200, 300, and 400 ms load imbalance injection, eager-SGD achieves 1.50x, 1.75x, and 2.01x speedup over synch-SGD, respectively. We observe that the more severe the load imbalance, the worse the performance of synch-SGD because of the synchronization overhead. On the other hand, the performance of eager-SGD is stable. Given that the throughput on a single GPU node with batch size of 2,048 is 0.64 steps/s, eager-SGD with 400 ms load imbalance injection still achieves 3.8x speedup in strong scaling on 8 GPU nodes. 

Fig.~\ref{hpregression-loss} presents the validation loss (mean squared error) as a function of the training time, which shows that eager-SGD using solo allreduce converges with equivalent loss value (around 4.7) to synch-SGD but significantly reduces the training time. Since the processes are not severely skewed and the stale gradients are added to the next training iteration (as discussed in Section~\ref{eagerSGDalgorithm}), using solo allreduce is enough for convergence. When using majority allreduce, the throughput of eager-SGD is lower than using solo allreduce (1.64 step/s vs 1.37 step/s with 200 ms load imbalance injection).

\subsubsection{ResNet-50 on ImageNet, Light Load Imbalance} 
\label{resnetevaluation}

Residual Network (ResNet)~\cite{he2016deep} is widely used in computer vision tasks. To evaluate the performance of eager-SGD, we use 64 processes with a total batch size of 8,192 to train ResNet-50 on ImageNet for 90 epochs. To simulate the load imbalance environment, we randomly select 4 processes out of the 64 processes at every training step to inject a certain amount of delay, according to the performance variability on Cloud machines discussed in Section~\ref{imbsys}. 

Fig.~\ref{tp-imgnet} presents the throughput comparison between synch-SGD (Horovod and Deep500) and eager-SGD using solo allreduce. With 300 and 460 ms load imbalance injection, eager-SGD achieves 1.25x and 1.29x speedup over Deep500, respectively; 1.14x and 1.27x speedup over Horovod, respectively. Given that the throughput of a single GPU node with batch size of 128 is 1.56 steps/s, eager-SGD running on 64 processes with 460 ms load imbalance injection still achieves 49.8x speedup in weak scaling.

Fig.~\ref{traintop1-imgnet} and Fig.~\ref{testtop1-imgnet} present the Top-1 train and test accuracy as a function of the training time, respectively. We train the model three times for each SGD, and obtain stable accuracy results. For top-1 accuracy, Deep500 achieves 79.1\% train accuracy and 75.7\% test accuracy, Horovod achieves 79.0\% train accuracy and 75.8\% test accuracy, while eager-SGD using solo allreduce achieves 78.4\% train accuracy and 75.2\% test accuracy on average over different load imbalance injections. Note that without model synchronization at every 10 epochs, the top-1 test accuracy of eager-SGD decreases to 74.1\%. For top-5 accuracy, synch-SGD achieves 92.6\% test accuracy, while eager-SGD using solo allreduce achieves 92.4\% test accuracy on average. The experimental results on ResNet-50 demonstrate that eager-SGD (solo) significantly improves the training speed without losing accuracy for deep neural networks in light load imbalance environment.

Table~\ref{gossip-based-sgds} presents the throughput comparison with the asynchronous centralized SGD and the gossip-based SGDs for ResNet-50 on ImageNet. We randomly select 4 processes out of the 64 processes at every training step and inject 460 ms delay for each selected process. Asynch-PS~\cite{tensorflow2015-whitepaper} is the asynchronous Parameter-Sever-based (centralized) SGD provided by TensorFlow. The throughput of Asynch-PS is the lowest among all the SGD variants because of the performance bottleneck on the sever. Compared with Asynch-PS, eager-SGD achieves 2.64$\times$ speedup. D-PSGD~\cite{syncring} and SGP~\cite{assran2018stochastic} are gossip-based SGDs, which do not use global collective communication primitives, such as Allreduce. Alternatively, each process only communicates with its neighbors (two neighbors for D-PSGD and SGP). However, all the processes need to finish the communications of the current step before going to the next step. The Overlap SGP~\cite{assran2018stochastic} can mitigate the synchronization effect by overlapping communication and computation. According to the parameter setup in~\cite{assran2018stochastic}, we configure SGP to use one step of gradient computation to overlap the communication, namely the communication synchronization is delayed by one step. Using communication topology optimizations~\cite{assran2018stochastic, asyncring}, each process can globally propagate its local update using $\mathcal{O}\left(\log P\right)$ steps. Note that eager-SGD only uses $\mathcal{O}\left(1 \right)$ step to globally propagate the local update. As shown in Table~\ref{gossip-based-sgds}, eager-SGD also outperforms the gossip-based SGDs because of the feature of asynchrony.

\begin{table}[!t]
  \caption{Throughput comparison with the asynchronous centralized SGD and the gossip-based SGDs for ResNet-50 on ImageNet, using 64 processes (total batch size 8,192) under load imbalance environment.}
  \label{gossip-based-sgds}
  \centering
  \small
  \begin{tabular}{c c c c c }
    \toprule
    SGDs & Asynch-PS~\cite{tensorflow2015-whitepaper} & D-PSGD~\cite{syncring} & SGP~\cite{assran2018stochastic} & eager-SGD \tabularnewline
    \midrule
    step/s & 0.45 & 0.94 & 1.02 & 1.19 \tabularnewline
    \bottomrule
  \end{tabular}
\end{table}

\subsubsection{ResNet-32 on Cifar-10, Severe Load Imbalance}

To test the robustness of eager-SGD, we train ResNet-32 on Cifar-10 with 8 processes for 190 epochs in a severe load imbalance environment. All 8 processes are skewed by injecting load imbalance from 50 ms to 400 ms at every training step. The injection amount over the processes is shifted after each step. Fig.~\ref{testtop1-cifar10} presents the test accuracy as a function of the training time. Eager-SGD using solo allreduce has the highest training speed but with lower test accuracy. Solo allreduce only waits for the fastest process to inform the other processes to participate in allreduce, but most of them will contribute stale gradients. Majority allreduce can solve the lower accuracy problem caused by solo allreduce, which achieves approximately equivalent accuracy to synch-SGD with 1.29x speedup. The results demonstrate that eager-SGD using majority allreduce is tolerant to severe load imbalance.
\begin{figure}[t!]
\centering\includegraphics[width=.95\linewidth]{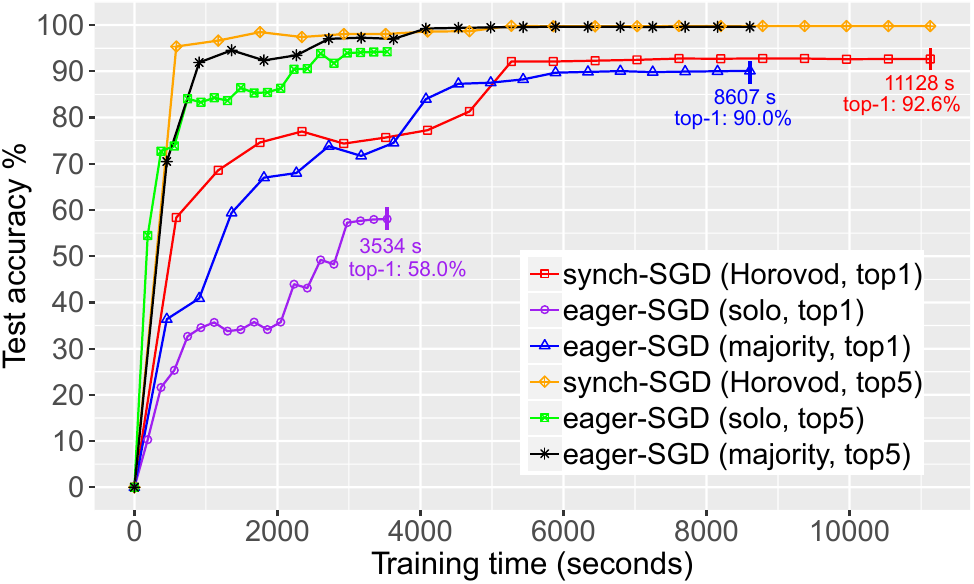}
\caption{\label{testtop1-cifar10} Top-1 test accuracy of synch-SGD (Horovod) and eager-SGD for ResNet-32 on Cifar-10 using 8 processes. Each point is at the boundary of every 10 epochs.}
\end{figure}

\subsection{Case Study: Video Classification}
\label{lstmevaluation}

As discussed in Section~\ref{imbvideomodel}, LSTM on UCF101 for video classification has inherent workload imbalance because of different workload for different batches. We use Inception v3~\cite{szegedy2016rethinking}, a CNN model, to extract a 2,048-wide feature from each frame of the videos, and then pass the sequences of features to an LSTM model. The training time reported in the paper is only for the LSTM model, not including the preprocessing time using Inception v3.

\begin{figure}[!t]
	\centering
	\begin{subfigure}{\linewidth}
		\centering\includegraphics[width=.92\linewidth]{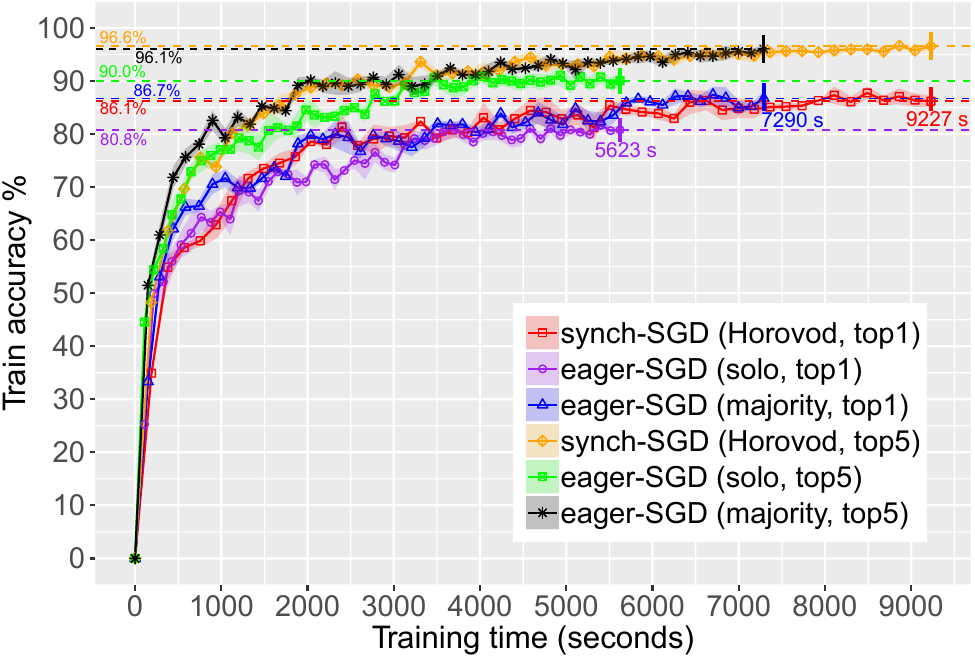}  
		\caption{Training accuracy.}
		\label{train-ucf101}
	\end{subfigure}
	\begin{subfigure}{\linewidth}
		\centering\includegraphics[width=.92\linewidth]{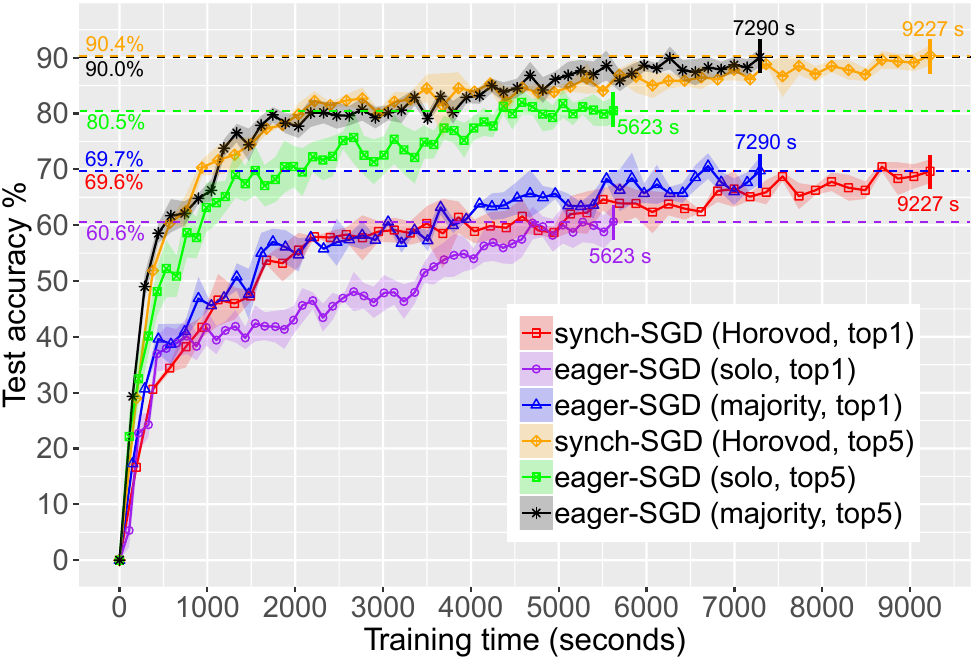}
		\caption{Test accuracy.}\label{test-ucf101}
	\end{subfigure}
	\caption{Training results for LSTM on UCF101 using 8 processes. Each point is at the boundary of one epoch.}
	\label{ucf101_eval}
\end{figure}

To evaluate eager-SGD, we use 8 processes with a total batch size of 128 to train LSTM on UCF101 for 50 epochs (more information is in Table~\ref{networks}). Fig.~\ref{train-ucf101} and Fig.~\ref{test-ucf101} present the train and test accuracy as a function of the training time, respectively. For each SGD, we train the model four times and plot the curves using the average values. Colored areas around the curves are confidence intervals with the boundaries representing the standard deviation. Although eager-SGD using solo allreduce achieves 1.64x speedup over Horovod, it has lower accuracy. Eager-SGD (solo) achieves on average 60.6\% (up to 70.4\%) top-1 test accuracy while Horovod achieves on average 69.6\%. This is because the workload of the video model is severely unbalanced, and solo allreduce introduces too many stale gradients. In contrast, eager-SGD using majority allreduce achieves 1.27x speedup over Horovod with equivalent accuracy. For example, Horovod achieves on average 69.6\% top-1 test accuracy (up to 72.2\%) and 90.4\% top-5 test accuracy (up to 91.9\%), while eager-SGD using majority allreduce achieves on average 69.7\% top-1 test accuracy (up to 72.8\%) and 90.0\% top-5 test accuracy (up to 91.7\%). Train accuracy results (in Fig.~\ref{train-ucf101}) show a similar trend as the test accuracy. Horovod achieves on average 86.1\% top-1 train accuracy and 96.6\% top-5 train accuracy, while eager-SGD using majority allreduce achieves on average 86.7\% top-1 train accuracy and 96.1\% top-5 train accuracy. All the accuracy results are consistent with that claimed in recent work~\cite{yue2015beyond}. The training speed and accuracy for Deep500 (not plotted in figures) are similar to Horovod. The results show that majority allreduce, with its statistical guarantee, is sufficient to both speed up training and achieve good accuracy.

The total training time for 50 epochs using a single GPU node with batch size of 16 and 128 is 34,301 seconds and 6,314 seconds, respectively. In weak scaling, Synch-SGD (Horovod) and eager-SGD using majority allreduce achieve 3.72x and 4.71x speedup on 8 GPU nodes, respectively. In strong scaling, synch-SGD and eager-SGD using majority allreduce do not have speedup on 8 GPU nodes; in contrast, eager-SGD using solo allreduce achieves 1.12x speedup on 8 GPU nodes in strong scaling, but with lower accuracy. Note that increasing the batch size can further improve the speedup in strong scaling for eager-SGD. However, large batch sizes commonly need carefully-tuned learning rate schedules to achieve good accuracy~\cite{you2018imagenet}, which is out of scope.

\section{Related Work}

\paragraph{Deep learning}

Parameter Server SGD implementations use synchronous~\cite{scaling14,rudra}, asynchronous~\cite{dean12,projadam}, stale-~\cite{ssp13,zhang2015staleness}, and approximate-synchronous~\cite{gaia} SGD, where the latter two limit the age and insignificance of the gradients, respectively. For synchronous Parameter Server SGD, communication granularity and scheduling optimizations~\cite{jayarajan2019priority} are studied to better overlap communication and computation. In a decentralized setting, asynchrony is achieved by performing communication on an explicit subset of the learners, e.g., forming a ring~\cite{asyncring} or a general graph~\cite{xie16sfb}; or on a random subset using Gossip Algorithms~\cite{gossipgrad,jin16}. These modes achieve some degree of asynchrony, but take $\mathcal{O}\left(P\right)$ or $\mathcal{O}\left(\log P\right)$ (for ring or gossip-based schemes, respectively) update steps to disseminate the information to all $P$ learners. To the best of our knowledge, this is the first work that implements asynchronous and stale-synchronous decentralized SGD where the messages propagate to all nodes in one step.

\paragraph{Collective communication}
Several algorithms can be used to implement allreduce operations, and the optimal algorithm depends on network topology, number of processes, and message size~\cite{thakur2005optimization}.
For large message sizes and large number of processes, practical implementations employ the ring-allreduce~\cite{gibiansky2017bringing} or the Rabenseifner's Algorithm~\cite{rabenseifner2004optimization}. Independently from the specific
algorithm, the semantic of the allreduce implies processes synchronization. With eager-SGD we relax this synchronization by using solo and majority allreduce operations.

\section{Conclusions}

In this work, we show that load imbalance is prevalent in deep learning problems with variable-length inputs, and increasingly becoming an issue in cloud systems. To that end, we propose eager-SGD: an asynchronous decentralized version of SGD where fast processes contribute gradients without waiting for slow processes. Using the resilience of machine learning to bounded error, we implement two variants of partial collectives --- solo and majority allreduce --- enabling this behavior without a central parameter server.

The experimental results reaffirm our theoretical analysis, showing that eager-SGD using solo allreduce speeds up the training process (1.29$\times$ and 2.64$\times$ faster than the synchronous decentralized SGD and the asynchronous centralized SGD on ImageNet, respectively) in light load imbalance environments. As the load imbalance increases, the convergence rate of solo allreduce degrades, in which case majority allreduce speeds up the training process (1.27$\times$ faster than the synchronous decentralized SGD on UCF101) yet desirable generalization.

The research can extend in different directions. Firstly, the promising results make eager-SGD attractive for other applications as well, such as language models and object detection. Secondly, in order to provide different quorum sizes, it is possible to construct a spectrum between solo, majority, and full collectives. Lastly, partial collectives can be used for other algorithms beyond SGD.

\begin{acks}

This project has received funding from the European Research Council (ERC) under the European Union’s
Horizon 2020 programme (grant agreement DAPP, No. 678880; grant agreement No. 805223, ERC Starting Grant ScaleML). We also would like to thank the Swiss National Supercomputing Center (CSCS) for providing the computing resources and for their excellent technical
support.

\end{acks}

\bibliographystyle{ACM-Reference-Format}
\bibliography{mybib}


\clearpage

\appendix
\section{Artifact Appendix}

\subsection{Abstract}

We provide source code of eager-SGD and scripts to run experiments from the paper. This artifact
is run on the Piz Daint supercomputer. This artifact
supports the paper by making it possible to reproduce the
figures and numbers in this paper, and it can be validated by
comparing the figures and results that this artifact’s scripts
generate with the data from the paper.

\subsection{Artifact check-list (meta-information)}

{\small
\begin{itemize}
  \item {\bf Algorithm:} Eager Stochastic Gradient Descent (eager-SGD)
  \item {\bf Compilation:} cmake, g++, Python 3.6
  \item {\bf Data set:} ImageNet, CIFAR-10, UCF101
  \item {\bf Run-time environment:} Cray MPICH 7.7.2, TensorFlow-gpu r1.11, Horovod, mpi4py, Keras
  \item {\bf Hardware:} Piz Daint Supercomputer (Intel Xeon E5-2690 CPU, NVIDIA Tesla P100 GPU)
  \item {\bf Execution:} SLURM job scheduler on Piz Daint
  \item {\bf Metrics:} Execution time, training throughput, loss values, Top1 and Top5 accuracy
  \item {\bf Output:} TXT files and Figures
  \item {\bf Experiments:} Use the provided scripts in the artifact to build, schedule jobs, and generate figures
  \item {\bf How much disk space required (approximately)?:} 400 GB
  \item {\bf How much time is needed to prepare workflow (approximately)?:} Assuming access to Piz Daint, 30 minutes
  \item {\bf How much time is needed to complete experiments (approximately)?:} About 90 hours if each job can be scheduled to run immediately. Considering the job queuing time, it may take one week.
\end{itemize}

\subsection{Description}

\subsubsection{How delivered}
Via Dropbox: \url{https://www.dropbox.com/s/3k8xgw0rh0s0m7j/eager-SGD-artifact.zip?dl=0}

\subsubsection{Hardware dependencies}
This artifact uses Piz Daint supercomputer.

\subsubsection{Software dependencies}
To run the experiments, Python 3.6, TensorFlow-GPU, Horovod, mpi4py, MPICH, and Keras are required. To plot out the figures, RStudio and ggplot2 are required.

\subsubsection{Data sets}
Download the training and validation images of ImageNet from \url{http://www.image-net.org/challenges/LSVRC/2010/downloads}

\noindent Download the binary version of CIFAR-10 from \url{https://www.cs.toronto.edu/\textasciitilde kriz/cifar-10-binary.tar.gz}

\noindent Download UCF101 from \url{https://www.crcv.ucf.edu/data/UCF101/UCF101.rar}

\subsection{Installation}

1. Download the artifact and move it to your personal \texttt{\$WORK} directory. Extract the artifact using \texttt{unzip}. 

\noindent 2. Install the dependent Python modules.

\texttt{> cd \$WORK/eager-SGD-artifact/eager-SGD}

\texttt{> pip install -r requirements.txt}

\noindent 3. Compile \textit{fflib2} and set the environment variable.

\texttt{> cd \$WORK/eager-SGD-artifact/eager-SGD/fflib2/lib}

\texttt{> cmake .. \&\& make}

\texttt{> export LD\_LIBRARY\_PATH=\$WORK/eager-SGD-artifact/\\eager-SGD/fflib2/lib:\$LD\_LIBRARY\_PATH}

\noindent 4. Configure a CMakelist file which will be used for compiling the customized Tensorflow operators.

\texttt{> vim \$WORK/eager-SGD-artifact/eager-SGD/deep500/\\deep500/frameworks/tensorflow/custom\_operators/\\CMakeLists.txt}

Update \texttt{include\_directories} and \texttt{link\_directories} by the path where \textit{fflib2} is installed. Set \texttt{TENSORFLOW\_PATH} by the path where TensorFlow is installed.

\texttt{> export PYTHONPATH=\$PYTHONPATH:\$WORK/eager-SGD-\\artifact/eager-SGD/deep500/}

\subsection{Experiment workflow}

To run experiments, users run the provided scripts that will schedule runs of the executable on Piz Daint via \texttt{sbatch}.
To run jobs on Piz Daint, one must put them on a queue and
wait until they are scheduled. Once these experiments finish, the results of execution time, loss values, Top1 and Top5 accuracy will be output. We provide scripts that will compile these
output results into TXTs using Python, and from these TXTs, we have included
scripts that will use R to create the figures that we used in the paper.

\subsection{Evaluation and expected result}
Users are expected to reproduce the results in this paper, specifically generating the figures in Section~\ref{eval}. Different versions of MPICH, TensorFlow-gpu, Horovod, and mpi4py may lead to slightly variant results compared with the numbers reported in the paper, but this does not affect the general trends claimed in the paper.

\ \newline
\noindent 1. Evaluate solo and majority allreduce and generate Fig.~\ref{eagerallreduce}.

\texttt{> cd \$WORK/eager-SGD-artifact/test-scripts/allreduce\\-scripts}

Submit the jobs. 

\texttt{> ./sbatch\_jobs.sh}

It may take about 10 minutes to finish the jobs, and then outputs \texttt{majority.txt}, \texttt{solo.txt}, and \texttt{mpi.txt}. Next, run the file \\\texttt{statistics\_summary.py} to read the output files and calculate the mean and standard deviation for the latency and results.

\texttt{> python statistics\_summary.py}

Now it should generate \texttt{latency\_results\_summary.txt}. Next, run the file \texttt{./plotRstudio/plot-Figure\ref{eagerallreduce}.R} (Rstudio and ggplot2 are required) to generate Fig.~\ref{eagerallreduce}. Detailed steps are stated in \\\texttt{./plotRstudio/ReadMe}.

The expected results are as follows:

a) For the average latency, \texttt{Solo\_Allreduce} < \texttt{Majority\_\\Allreduce} < \texttt{MPI\_Allreduce}, where "<" means "less than".

b) For the average results, \texttt{Solo\_Allreduce} < \texttt{Majority\_\\Allreduce} < \texttt{MPI\_Allreduce}, where "<" means "less than".

\ \newline
\noindent 2. Train hyperplane regression and generate Fig.~\ref{hpregression}.

\texttt{> cd ./\$WORK/eager-SGD-artifact/test-scripts/\\hyperplane-scripts}

Submit the jobs.

\texttt{> ./sbatch\_jobs.sh}

It may take about 20 minutes to finish the jobs, and then outputs \texttt{solo200.txt}, \texttt{solo300.txt}, \texttt{solo400.txt}, \texttt{dfive200.txt},\\ \texttt{dfive300.txt}, and \texttt{dfive400.txt}. Next, run the file \texttt{statistics\\\_summary.py} to  summarize the throughput, runtime, and loss data.

\texttt{> python statistics\_summary.py}

Now it should generate \texttt{throughput\_summary.txt} and \texttt{loss\_\\summary.txt}. Next, run the file \texttt{./plotRstudio/plot-Figure\ref{hpregression}.R} to plot Fig.~\ref{hpregression}. Detailed steps are stated in \texttt{./plotRstudio/ReadMe}.

The expected results are as follows:

a) For the throughput of training, eager-SGD achieves higher throughput than synch-SGD.

b) Eager-SGD converges and achieves approximately equivalent loss value to synch-SGD.   

\ \newline
\noindent 3. Train ResNet-50 on ImageNet and generate Fig.~\ref{acc-imagenet}.

Generate the TensorFlow data format for ImageNet, which may take several hours.

\texttt{> python \$WORK/eager-SGD-artifact/test-models/tf-\\models-r1.11/official/data/build\_imagenet\_data.py}

\texttt{> cd \$WORK/eager-SGD-artifact/test-scripts/imagenet\\-scripts}

Copy \texttt{synchm.sh} to the checkpoint directory. Submit the jobs. 

\texttt{> ./sbatch\_jobs.sh}

It may take about tens of hours to finish the jobs. It should generate \texttt{solo300.txt}, \texttt{solo460.txt}, \texttt{hvd300.txt}, \texttt{hvd460.txt}, \texttt{dfive300.txt}, and \texttt{dfive460.txt}, which contain the output data of the jobs. Make sure all the jobs have been finished. Then, run the file \texttt{statistics\_summary.py} to read the output files and summarize the throughput, runtime, and accuracy data.

\texttt{> python statistics\_summary.py}

It generates \texttt{imgnetthroughput\_64p.txt}, \texttt{top1testimgnet\_\\64p\_runtime.txt} and \texttt{top1trainimgnet\_64p\_runtime.txt}. Next, run the file \texttt{./plotRstudio/plot-Figure\ref{acc-imagenet}.R} to generate Fig.~\ref{acc-imagenet}. Detailed steps are stated in \texttt{./plotRstudio/ReadMe}.

The expected results are as follows:

a) For the throughput of training, eager-SGD achieves higher throughput than synch-SGD (Horovod and Deep500).

b) For the accuracy, eager-SGD converges and achieves approximately equivalent accuracy value to synch-SGD.

\ \newline
\noindent 4. Train ResNet-32 on CIFAR-10 and generate Fig.~\ref{testtop1-cifar10}.

\texttt{> cd \$WORK/eager-SGD-artifact/test-scripts/cifar10\\-scripts}

Copy \texttt{synchm.sh} to the checkpoint directory. Submit the jobs. 

\texttt{> ./sbatch\_jobs.sh}

It may take about several hours to finish the jobs, and then outputs \texttt{hvd.txt}, \texttt{major.txt}, and \texttt{solo.txt}. Make sure all the jobs have been finished. Then, run the file \texttt{statistics\_summary.py} to read the output files and summarize the runtime and accuracy data.

\texttt{> python statistics\_summary.py}

It should generate \texttt{cifar10\_accuracy\_runtime.txt}. Next, run the file \texttt{./plotRstudio/plot-Figure\ref{testtop1-cifar10}.R} to generate Fig.~\ref{testtop1-cifar10}. Detailed steps are stated in \texttt{./plotRstudio/ReadMe}.

The expected results are as follows:

a) For the runtime of training, eager-SGD (solo) < eager-SGD (majority) < synch-SGD (horovod), where "<" means "less than".

b) For the accuracy, eager-SGD (majority) is very close to synch-SGD (horovod); however, the accuracy of eager-SGD (solo) is apparently lower than eager-SGD (majority) and synch-SGD (horovod).

\ \newline
\noindent 5. Train LSTM on UCF101 and generate Fig.~\ref{ucf101_eval}.

Extracting features from the raw data. It may take several hours.

\texttt{> python \$WORK/eager-SGD-artifact/test-models/lstm-\\video-classification/extract\_features.py}

Submit the jobs.

\texttt{> cd \$WORK/eager-SGD-artifact/test-scripts/lstm-\\scripts}

\texttt{> ./sbatch\_jobs.sh}

It may take several hours to finish the jobs. It should generate \texttt{hvd.iter1-4.txt}, \texttt{solo.iter1-4.txt}, and \texttt{majority.iter1-\\4.txt}, which contain the output data. Make sure all the jobs have been finished. Then, run the file \texttt{statistics\_summary.py} to read the output files and summarize the runtime and accuracy data.

\texttt{> python statistics\_summary.py}

It should generate \texttt{test\_statistics\_summary.txt} and \texttt{train\_\\statistics\_summary.txt}. Next, run the file \texttt{./plotRstudio/\\plot-Figure\ref{ucf101_eval}.R} to generate Fig.~\ref{ucf101_eval}. Detailed steps are stated in \\\texttt{./plotRstudio/ReadMe}.

The expected results are as follows:

a) For the runtime of training, eager-SGD (solo) < eager-SGD (majority) < synch-SGD (horovod), where "<" means "less than".

b) For the accuracy, eager-SGD (majority) is very close to synch-SGD (horovod); however, the accuracy of eager-SGD (solo) is lower than eager-SGD (majority) and synch-SGD (horovod).

\subsection{Experiment customization}
Users can modify the scripts in the subdirectories of \texttt{\$WORK/eager\\-SGD-artifact/test-scripts} to customize the experiments.

Modify the line of \texttt{\#SBATCH --nodes=<number-of-nodes>} to change the number of nodes (processes). Modify the input parameter \texttt{-bs=<batch-size-per-node>} to change the batch size per node. Note that the \texttt{total-batch-size = number-of-nodes * batch-size-per-node}, which means the change of the number of nodes and the batch size per node would change the total batch size. However, different total batch sizes may lead to different results for the train and test accuracy. 

Modify the line of \texttt{\#SBATCH --time=<time-quota>} to change the time limit to run the job.

Users can also use eager-SGD to train other TensorFlow-based models that are not listed in the paper. To achieve this, simply wrap the TensorFlow optimizer using the eager-SGD optimizer, and then use the eager-SGD optimizer instead to train the model.

\subsection{Notes}
Some jobs may not be scheduled to run for a surprising long time because of the busy use of the machine. In this case, they may be automatically cancelled by the job scheduler. This is usually resolved by rescheduling the cancelled jobs using \texttt{sbatch}. 

\subsection{Methodology}

Submission, reviewing and badging methodology:\\
\url{http://cTuning.org/ae/submission-20190109.html}\\
\url{http://cTuning.org/ae/reviewing-20190109.html}\\
\url{https://www.acm.org/publications/policies/artifact-review-badging}


\end{document}